\documentclass[3p,times]{elsarticle}
\journal{Theoretical Computer Science}

\usepackage{amsmath}
\usepackage{amssymb}
\usepackage{amsthm}
\usepackage{booktabs}
\usepackage{tabularx}
\usepackage{xltabular}
\usepackage{tikz}
\usetikzlibrary{arrows.meta,positioning}

\usepackage[hidelinks]{hyperref}
\usepackage{orcidlink}

\theoremstyle{plain}
\newtheorem{theorem}{Theorem}
\newtheorem{lemma}{Lemma}
\newtheorem{corollary}{Corollary}
\newtheorem{proposition}{Proposition}

\theoremstyle{definition}
\newtheorem*{definition*}{Definition}
\newtheorem{example}{Example}

\theoremstyle{remark}
\newtheorem{remark}{Remark}

\begin{document}

\begin{frontmatter}

\title{Certificates for short extending words in a finite automaton}

\author[a]{Michele Miccinesi\,\orcidlink{0000-0003-3138-8063}}
\ead{mih.ele@gmail.com}
\affiliation[a]{organization={University of Pisa},
            country={Italy}}

\begin{abstract}
Let $\mathcal A$ be a complete deterministic finite automaton on a state set
$Q$ of size $n$ with $k$ letters, and for a proper nonempty subset
$S$ of $Q$ let $\mathrm{minext}(S)$ be the length of a shortest word $u$ with
$|Su^{-1}|>|S|$, where $Su^{-1}=\{q: q\cdot u\in S\}$.  
To each state $q$ attach the integer
$\beta^{\ast}_q=\sum_{t=1}^{n-1}k^{\,n-1-t}(\mathrm{indeg}_t(q)-k^{t})$, where
$\mathrm{indeg}_t(q)$ counts the pairs $(p,u)$ with $|u|=t$ and $p\cdot u=q$,
and let $B(S)=\sum_{q\in S}\beta^{\ast}_q$.  On every synchronizing
automaton, $B(S)\ge0$ implies
$\mathrm{minext}(S)\le n-1$, so, as $B(Q)=0$, one of $S$ and $Q\setminus S$
extends within $n-1$; when $B(S)>0$ no hypothesis is needed.  Kari's
Eulerian extension lemma is the case $\beta^{\ast}=0$, and
$\beta^{\ast}$, like every member of the family
$\sum_{t=1}^{n-1}c_t\sigma_t$, $c_t>0$, vanishes identically if and only
if the automaton is Eulerian, where
$\sigma_t(S)=\sum_{q\in S}(\mathrm{indeg}_t(q)-k^{t})$.
On strongly connected automata $\sigma_t(S)/k^{t}$ has Ces\`aro limit
$n\,e(S)/e(Q)-|S|$ for Friedman's weight $e$; that limit certifies
singletons but no larger subset in general.
The hypothesis $B(S)\ge0$ cannot be relaxed by one integer unit, nor
can the constant $n-1$ be improved.
A second-moment test on the sizes $|Su^{-1}|$ certifies 60 to 95 percent
of the subsets with $B(S)<0$ at $n\le7$.  Along non-Eulerian automata
whose words of length $n-1$ merge a fraction of the state pairs bounded
below, with $\max_q\mathrm{indeg}_{n-1}(q)=o(nk^{n-1})$, it certifies
all but a vanishing share of them.  The functional $B$ certifies half
of the subsets outside $\{B=0\}$.
At each subset size coprime to $n$ ($n\ge4$) some synchronizing Eulerian
binary automaton attains the constant $n-1$; whether only there is open.
No reset bound follows: \v{C}ern\'y's automata have subsets not extending
within $n-1$.
\end{abstract}


\end{frontmatter}

\section{Introduction}
\label{sec:intro}

\v{C}ern\'y's conjecture asks whether every synchronizing $n$-state automaton has a
reset word of length at most $(n-1)^{2}$ \cite{ce}.  One of the two standard
approaches to a quadratic bound is the extension method: if every subset $S$ of the
state set $Q$ of size at least $2$ and at most $n-1$ has an extending word of length
at most $\alpha n$, that is a word $u$ with $|Su^{-1}|>|S|$, then the
automaton resets within $1+\alpha n(n-2)$ \cite[Proposition~1]{vo}.  
The hypothesis is uniform: it must hold for every subset.

This paper exhibits an
explicit integer vector $\beta^{\ast}\in\mathbb Z^{Q}$, computable in
$O(kn^{2})$ operations from the transition table, such that on every
synchronizing complete automaton every proper nonempty subset $S$ with
$B(S)=\sum_{q\in S}\beta^{\ast}_q\ge0$ has an extending word of length
at most $n-1$ (Theorem~\ref{thm:main}).  In the words of the title, the
sign of the linear functional $B$ is a certificate for extension within
$n-1$: a sufficient condition for one subset, read off from $\beta^{\ast}$
in $|S|$ additions, with no requirement on the other subsets.  The entry at
$q$ is $\beta^{\ast}_q=\sum_{t=1}^{n-1}k^{\,n-1-t}(\mathrm{indeg}_t(q)-k^{t})$,
where $\mathrm{indeg}_t(q)$ counts the pairs $(p,u)$ with $|u|=t$ and
$p\cdot u=q$: the $t$-th term is the excess of $q$ over the mean $t$-step
in-degree $k^{t}$, and $\beta^{\ast}$ is a weighted sum of these excesses;
summed over $S$, the excess at length $t$ is the defect
$\sigma_t(S)=\sum_{q\in S}(\mathrm{indeg}_t(q)-k^{t})$, and
$B(S)=\sum_{t=1}^{n-1}k^{\,n-1-t}\sigma_t(S)$.
Call a proper nonempty subset with $\mathrm{minext}(S)>n-1$ stuck.  On a
synchronizing automaton every stuck subset therefore has $B(S)<0$: the
stuck subsets lie on one side of one hyperplane of $\mathbb R^{Q}$,
intersected with the cube.  No
reset bound follows, nor could it in general (Section~\ref{sec:no-reset}).  
The functional $B$ is the member $T=n-1$ of the family
$B_T=\sum_{t=1}^{T}k^{\,T-t}\sigma_t$, each member certifying extension
within $T$ wherever it is strictly positive (Section~\ref{sec:certificate}).
The synchronizing hypothesis is used only on the subsets with every
$\sigma_t(S)=0$, a linear subspace intersected with the cube inside
$\{B=0\}$ (Theorem~\ref{thm:krylov}); on 55 to 93 percent of the
non-Eulerian automata measured (\ref{app:void-share}) that subspace
contains no proper nonempty subset.  Section~\ref{sec:qcert} presents a
second-moment test on the preimage sizes $|Su^{-1}|$, which reduces the
set of uncertified subsets: over the populations measured it establishes
$\mathrm{minext}(S)\le n-1$ on 60.1 to 95.5 percent of the subsets with
$B(S)<0$, at the expense of linearity, $O(kn^{3})$ against $O(kn^{2})$ for all
lengths $t\le n-1$ and $O(2^{n}n)$ against $O(2^{n})$ over all subsets.
Section~\ref{sec:asymptotic} bounds both tests on every automaton.  Call
the pairs $(p,u)$ with $|u|=n-1$ the landings, received at $p\cdot u$.
Along any sequence of non-Eulerian automata whose words of length $n-1$
merge a fraction of the state pairs bounded below, with the largest
fraction of the landings at one state vanishing, the second-moment test
certifies all but a vanishing share of the subsets with $B(S)<0$; on
every automaton $B$ certifies half of the subsets outside $\{B=0\}$.
Along such a sequence the uncertified share vanishes exponentially fast
in $n$ when no state receives more than a bounded multiple of the mean
number of landings.  The second-moment test certifies no stuck subset,
so the stuck subsets are part of that uncertified share.  The merging hypothesis is also necessary: for $n\ge100$, where the
words of length $n-1$ merge fewer pairs on average than $n^{3/2}/425$
times the standard deviation over the states of the number of landings
divided by its mean, more than
$2.7$ percent of the $2^{n}$ subsets fail the test at length $n-1$
(Proposition~\ref{prop:barrier}).

The starting point is Kari's theorem that \v{C}ern\'y's conjecture holds on
Eulerian automata, where every state has in-degree $k$ \cite{ka}.
The vector $\beta^{\ast}$ vanishes identically on the Eulerian automata and nowhere else
(Theorem~\ref{thm:fibre}), 
where Theorem~\ref{thm:main} degenerates to Kari's own extension lemma.
On every other automaton $\beta^{\ast}\ne0$, and Theorem~\ref{thm:main} replaces the Eulerian
hypothesis by the inequality $B(S)\ge0$, one subset at a time.

What the certificate adds to Kari's lemma is twofold.  It replaces the
Eulerian hypothesis by a per-subset inequality whose functional vanishes
identically exactly on the Eulerian automata,
for every strictly positive weighting of the word lengths $t=1,\dots,n-1$.  
Each member of the family is, moreover, a polynomial-time computable
per-subset certificate for growth in cardinality under an arbitrary
in-degree profile.

The counting identity behind the certificate
is Kari's one-step identity applied to the $t$-th power automaton; the chain
argument in the proof of Theorem~\ref{thm:main}, Lemma~\ref{lem:stabilise},
is Kari's Lemma 3 in the form Steinberg gave it; and the passage to
power automata is Steinberg's averaging lemma.  
Section~\ref{sec:previous} traces each ingredient to its source.  What
the paper builds from them is the pair above: the exactness of the
degeneration, a property of the whole positive cone of weightings
(Theorem~\ref{thm:weights}), and the per-subset certificate reading.
The family also has a limit: on strongly connected automata the
Ces\`aro means of the normalised defects $\sigma_t(S)/k^{t}$ converge to
Friedman's weight, the stationary vector of Berlinkov's bound, centred
(Proposition~\ref{prop:cesaro}).  That functional is
linear, integer-valued, antisymmetric and zero exactly on the Eulerian
automata, like $B$, yet it certifies only singletons
(Section~\ref{sec:stationary}), so the limit of the family lies outside
the family.

Kari's concluding section argues that imbalance in the in-degrees helps
synchronization, and that the
digraphs likely to be hard for a general proof are the nearly, but not
fully, balanced ones.  Theorem~\ref{thm:main} is a statement about that
region: a subset is certified when $B(S)\ge0$, whatever the deviation
$d=\sum_q|\mathrm{indeg}(q)-k|$ of the in-degrees from $k$.

The paper is organised as follows.  Section~\ref{sec:setting} fixes
notation.  Section~\ref{sec:certificate} derives the certificate from a
counting identity and gives its closed form.  Section~\ref{sec:main-theorem}
proves the main theorem and localises the only use of synchronization.
Section~\ref{sec:antisym-fibre} proves the antisymmetry $B(Q\setminus S)=-B(S)$, that
$\beta^{\ast}$ vanishes identically exactly on the Eulerian automata, for
every weighting, and that for every $n\ge4$ and every subset size
coprime to $n$ some synchronizing Eulerian binary automaton attains the
constant $n-1$ at that size (Theorem~\ref{thm:gcd}); whether only at
such sizes is open.  Section~\ref{sec:sharpness}
establishes sharpness, introduces the second-moment certificate and
bounds the shares of the two certificates as the state set grows.
Section~\ref{sec:previous} relates the results to the work of Kari,
Steinberg and Berlinkov, its Section~\ref{sec:stationary} settling what the
stationary-vector functional certifies.
Section~\ref{sec:notdo} states three negative results: no reset bound
follows, $B$ is not monotone along flat moves $S\mapsto Sx^{-1}$, those
with $|Sx^{-1}|=|S|$, and a walk of flat moves to a subset with an
extending letter need not exist.  The
appendices hold the supporting material: \ref{sec:verification} the
populations and the verification, \ref{app:fibre-proofs} the proofs of
Theorem~\ref{thm:gcd} and Proposition~\ref{prop:gcd-partial} on the
Eulerian automata, \ref{app:stationary} the proofs, witnesses and
censuses for the stationary-vector functional, and \ref{app:qcert} the
measurements for the second-moment test.  \ref{sec:bounded-deviation}
holds the family of Kisielewicz and Szyku\l{}a, of deviation $d=n-3$, on
which no constant $\alpha$ bounds the extending length by $\alpha n$.

\section{Setting}
\label{sec:setting}

Throughout, $\mathcal A=(Q,\Sigma,\delta)$ denotes a complete deterministic finite automaton:
$|Q|=n\ge2$, $k=|\Sigma|\ge1$, and $\delta$ total.  We write $q\cdot u$ for
the action of a word $u\in\Sigma^{*}$ and, for $S\subseteq Q$,
\[Su^{-1}=\{q\in Q:\ q\cdot u\in S\},\]
\[\mathrm{minext}(S)=\min\{|u|:\ |Su^{-1}|>|S|\}\in\mathbb N\cup\{\infty\}.\]

The automaton is synchronizing if $|Q\cdot r|=1$ for some word $r$; 
$r$ is a reset word, and the reset threshold $\mathrm{rt}(\mathcal A)$ is the
length of a shortest one.  Throughout the paper, a proper nonempty subset
with $\mathrm{minext}(S)>n-1$ is called stuck; 
the certificate below certifies extension within $n-1$.

Matrix conventions: $[S]$ is the characteristic row vector of $S$,
$\pi(u)_{q,q'}=[\,q\cdot u=q'\,]$, $\mathbf1$ is the all-ones column vector,
and
\[M=\sum_{x\in\Sigma}\pi(x),\qquad M_{q,q'}=\#\{x\in\Sigma:\ q\cdot x=q'\}.\]
Completeness alone gives $M\mathbf1=k\mathbf1$ and
$(\mathbf1^{\mathsf T}M)_{q'}=\mathrm{indeg}(q')$.  Define the deviation vector
\begin{equation}
\label{eq:identity}
w:=\mathrm{indeg}-k\cdot\mathbf1,\qquad
\mathbf1^{\mathsf T}M=k\,\mathbf1^{\mathsf T}+w^{\mathsf T},
\end{equation}
so that $\sum_qw_q=0$ always.  The deviation of $\mathcal A$ is
$d:=\sum_q|w_q|$; $d=0$ is the Eulerian case.  Since the entries of
$w$ sum to zero, $d=2\sum_q\max(0,w_q)$ is even at every $n$ and every
alphabet size.

The only input used below is the identity~\eqref{eq:identity}, which holds for every complete
automaton at every alphabet size. 
Nothing below assumes anything about the shape of $w$.  
Some sections report exhaustive machine enumerations;
see \ref{sec:conventions} for the conventions.

\section{The certificate}
\label{sec:certificate}

For $t\ge1$ let
$\mathrm{indeg}_t(q)=\#\{(p,u):\ |u|=t,\ p\cdot u=q\}=(\mathbf1^{\mathsf T}M^{t})_q$
be the $t$-step in-degree.  Since $\sum_q\mathrm{indeg}_t(q)=k^{t}n$, the
mean $t$-step in-degree is $k^{t}$, and
$\mathrm{indeg}_t(q)-k^{t}$ is the excess of $q$ at length $t$; at $t=1$ it
is $w_q$.

The certificate can be deduced by counting one set of pairs in two ways.  
Fix $t\ge1$ and $S\subseteq Q$, and count the pairs $(p,u)$ with $|u|=t$ and
$p\cdot u\in S$.  
Grouping by the landing state gives
$\sum_{q\in S}\mathrm{indeg}_t(q)$; grouping by the word gives
$\sum_{|u|=t}|Su^{-1}|$.  Subtracting $k^{t}|S|$ from both sides, once as a
total and once distributed over the $k^{t}$ words of length $t$, yields an
identity for the defect at each length.

\begin{definition*}
For $t\ge1$ and $S\subseteq Q$ put
$\sigma_t(S)=\mathrm{indeg}_t(S)-k^{t}|S|$, where
$\mathrm{indeg}_t(S)=\sum_{q\in S}\mathrm{indeg}_t(q)$.
\end{definition*}

\begin{lemma}
\label{lem:defect}
For every complete automaton, every $S\subseteq Q$ and every
$t\ge1$,
\begin{equation}
\label{eq:defect}
\sigma_t(S)=\sum_{|u|=t}\bigl(|Su^{-1}|-|S|\bigr).
\end{equation}
\end{lemma}

\begin{proof}
$|Su^{-1}|=\sum_{q\in S}|qu^{-1}|$ and
$\sum_{|u|=t}|qu^{-1}|=\mathrm{indeg}_t(q)$, so
$\sum_{|u|=t}|Su^{-1}|=\mathrm{indeg}_t(S)$; subtract $k^{t}|S|$ and
distribute it over the $k^{t}$ words of length $t$.
\end{proof}

\begin{corollary}
\label{cor:strict}
For any complete automaton, any $S\subseteq Q$ and any
$t\ge1$,
\[\sigma_t(S)>0\ \Longrightarrow\ \mathrm{minext}(S)\le t.\]
\end{corollary}

\begin{proof}
If $\sigma_t(S)>0$, some summand of~\eqref{eq:defect} is positive.
\end{proof}

Corollary~\ref{cor:strict} assumes nothing beyond completeness: no synchronization, no
connectivity, no condition on the degrees or on the alphabet.  
Identity~\eqref{eq:defect}
thus splits the extension hypothesis into a strict part, some
$\sigma_t(S)>0$, and a boundary part, every $\sigma_t(S)=0$; we call these
the strict and boundary halves of the method.

Since the test at length $t$ certifies extension within $t$, and the
tests at different lengths may be aggregated with any positive weights
without losing the conclusion, the corollary is graded by the length $t$.    
The choice below admits a closed form.

\begin{definition*}
For $T\ge1$ put
$B_T(S)=\sum_{t=1}^{T}k^{\,T-t}\sigma_t(S)$, and write
$\beta^{\ast}_q=B_{n-1}(\{q\})$, that is
\[\beta^{\ast}_q=\sum_{t=1}^{n-1}k^{\,n-1-t}\bigl(\mathrm{indeg}_t(q)-k^{t}\bigr),
\qquad B(S):=B_{n-1}(S)=\sum_{q\in S}\beta^{\ast}_q.\]
\end{definition*}

The functional $B_T$ is linear with integer values, and $B_T(Q)=0$ because every term
contributes $k^{\,T-t}(k^{t}n-k^{t}n)=0$.  If $B_T(S)>0$ then, the weights being strictly positive, some
$\sigma_t(S)>0$ with $t\le T$, whence
Corollary~\ref{cor:strict} gives $\mathrm{minext}(S)\le T$.  Each member $B_T$ of the family
certifies extension within its own length $T$.

\begin{lemma}
\label{lem:closed}
For every $T\ge1$,
\begin{equation}
\label{eq:closed}
B_T^{\mathsf T}\;=\;\sum_{j=0}^{T-1}(T-j)\,k^{\,T-1-j}\,w^{\mathsf T}M^{\,j},
\end{equation}
where $B_T^{\mathsf T}$ denotes the row vector with entries $B_T(\{q\})$.  In
particular $\beta^{\ast\mathsf T}=\sum_{j=0}^{n-2}(n-1-j)k^{\,n-2-j}w^{\mathsf T}M^{\,j}$.
\end{lemma}

\begin{proof}
Iterating~\eqref{eq:identity},
$\mathbf1^{\mathsf T}M^{t}=k^{t}\mathbf1^{\mathsf T}+\sum_{j<t}k^{\,t-1-j}w^{\mathsf T}M^{\,j}$
(induction on $t$: multiply on the right by $M$ and substitute~\eqref{eq:identity}).  Hence
$\mathrm{indeg}_t-k^{t}\mathbf1^{\mathsf T}=\sum_{j<t}k^{\,t-1-j}w^{\mathsf T}M^{\,j}$
as row vectors, and summing over $t=1,\dots,T$ with weight $k^{\,T-t}$
collects the coefficient of $w^{\mathsf T}M^{\,j}$ as
$\sum_{t=j+1}^{T}k^{\,T-t}k^{\,t-1-j}=(T-j)k^{\,T-1-j}$.
\end{proof}

The definition of $\beta^{\ast}$ ranges
over $k+k^{2}+\dots+k^{n-1}$ words; \eqref{eq:closed} evaluates it with $n-2$ products
$v\mapsto vM$, each $O(kn)$ integer additions, so $\beta^{\ast}$ costs
$O(kn^{2})$ additions on numbers of $O(n\log k+\log n)$ bits, and $B_T$ costs
$O(kTn)$.  Evaluating $B(S)$ then costs $O(|S|)$, and $B$ on all $2^{n}$ subsets costs
$O(2^{n})$ by subset-sum recursion.  The closed form makes a certificate whose definition ranges over
exponentially many words computable in polynomial time.

\begin{example}
\label{ex:cerny}
Take the \v{C}ern\'y automaton $C_3$ on $Q=\{0,1,2\}$ (Figure~\ref{fig:cerny}):
the letter $a$
acts as the cycle $0\to1\to2\to0$ and the letter $b$ sends $0$ to $1$ and
fixes $1$ and $2$.  The in-degrees are $(1,3,2)$, so $w=(-1,1,0)$, and one
product with $M$ gives $w^{\mathsf T}M=(0,-1,1)$; by Lemma~\ref{lem:closed},
$\beta^{\ast\mathsf T}=2k\,w^{\mathsf T}+w^{\mathsf T}M=(-4,3,1)$, which sums
to $0$ as it must.  The subset $S=\{1,2\}$ has $B(S)=4>0$; since
$\sigma_1(S)=w_1+w_2=1>0$, some letter extends it: $b$ does, with
$|Sb^{-1}|=3$.  The singleton $\{0\}$ has $\beta^{\ast}_0=-4<0$ and is
stuck.  No word of length at most $2$ extends it, as the six nonempty words
of length at most $2$ confirm by hand.  The subset $\{0,2\}$ has $B=-3<0$,
so Theorem~\ref{thm:main} does not apply to it, though it extends within
$2$, by the word $ba$.  So $\{B<0\}$ is nonempty on $C_3$.
Figure~\ref{fig:cube} shows all eight subsets at once.
\end{example}

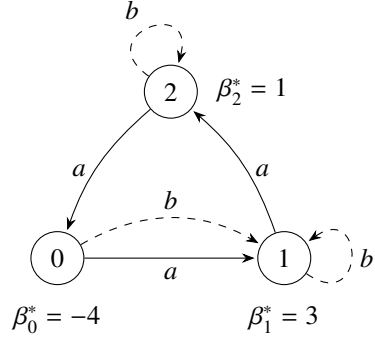
\begin{figure}[t]
\centering
\begin{tikzpicture}[>=Stealth,shorten >=1pt,
  state/.style={circle,draw,minimum size=7mm,inner sep=0pt}]
  \node[state] (q0) at (0,0) {$0$};
  \node[state] (q1) at (3,0) {$1$};
  \node[state] (q2) at (1.5,2.2) {$2$};
  \node[below=1.5mm of q0] {$\beta^{\ast}_0=-4$};
  \node[below=1.5mm of q1] {$\beta^{\ast}_1=3$};
  \node[right=1.5mm of q2] {$\beta^{\ast}_2=1$};
  \draw[->] (q0) -- node[below] {$a$} (q1);
  \draw[->] (q1) to[bend right=15] node[right,pos=.45] {$a$} (q2);
  \draw[->] (q2) to[bend right=15] node[left,pos=.55] {$a$} (q0);
  \draw[->,dashed] (q0) to[bend left=30] node[above] {$b$} (q1);
  \draw[->,dashed] (q1) to[out=-35,in=35,looseness=6] node[right] {$b$} (q1);
  \draw[->,dashed] (q2) to[out=145,in=75,looseness=6] node[above left] {$b$} (q2);
\end{tikzpicture}
\caption{The \v{C}ern\'y automaton $C_3$ of Example~\ref{ex:cerny}, with
$\beta^{\ast}$ attached to the states.  Solid arrows are the letter $a$,
dashed arrows the letter $b$.}
\label{fig:cerny}
\end{figure}

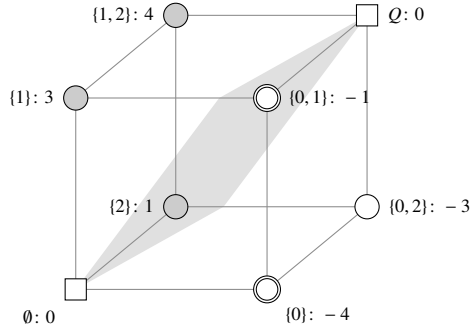
\begin{figure}[t]
\centering
\begin{tikzpicture}[scale=1.15,
  cert/.style={circle,draw,fill=black!20,minimum size=3.2mm,inner sep=0pt},
  neg/.style={circle,draw,fill=white,minimum size=3.2mm,inner sep=0pt},
  negstuck/.style={circle,draw,double,fill=white,minimum size=3.2mm,inner sep=0pt},
  triv/.style={rectangle,draw,fill=white,minimum size=2.8mm,inner sep=0pt},
  lab/.style={font=\scriptsize}]
  \fill[black!12] (0,0) -- (1.65,2.2) -- (3.35,3.15) -- (1.70,0.95) -- cycle;
  \draw[black!45] (0,0) -- (2.2,0) -- (2.2,2.2) -- (0,2.2) -- cycle;
  \draw[black!45] (1.15,0.95) -- (3.35,0.95) -- (3.35,3.15) -- (1.15,3.15) -- cycle;
  \draw[black!45] (0,0) -- (1.15,0.95);
  \draw[black!45] (2.2,0) -- (3.35,0.95);
  \draw[black!45] (0,2.2) -- (1.15,3.15);
  \draw[black!45] (2.2,2.2) -- (3.35,3.15);
  \node[triv,label={[lab]below left:{$\emptyset\colon0$}}] at (0,0) {};
  \node[negstuck,label={[lab]below right:{$\{0\}\colon-4$}}] at (2.2,0) {};
  \node[cert,label={[lab]left:{$\{1\}\colon3$}}] at (0,2.2) {};
  \node[cert,label={[lab]left:{$\{2\}\colon1$}}] at (1.15,0.95) {};
  \node[negstuck,label={[lab]right:{$\{0,1\}\colon-1$}}] at (2.2,2.2) {};
  \node[neg,label={[lab]right:{$\{0,2\}\colon-3$}}] at (3.35,0.95) {};
  \node[cert,label={[lab]left:{$\{1,2\}\colon4$}}] at (1.15,3.15) {};
  \node[triv,label={[lab]right:{$Q\colon0$}}] at (3.35,3.15) {};
\end{tikzpicture}
\caption{The eight subsets of $C_3$ as the cube $\{0,1\}^{Q}$, membership
of $0$, $1$ and $2$ along the three axes, with the value of $B$ at every
vertex and the plane $B=0$ shaded.  Filled circles are certified subsets
($B>0$), open circles have $B<0$, and the two squares
are the trivial subsets, which lie on the plane.  Antipodal vertices carry
opposite values, so every complementary pair of proper subsets has at least
one member with $B\ge0$ (Corollary~\ref{cor:antisym}); on this automaton no
proper subset lies on the plane, so here it is exactly one.  The two stuck
subsets, $\{0\}$ and $\{0,1\}$ (doubled border), both have $B<0$, the
larger of them $B=-1$.}
\label{fig:cube}
\end{figure}

\section{The main theorem}
\label{sec:main-theorem}

Corollary~\ref{cor:strict} says
nothing about a subset all of whose defects vanish.  Under
synchronization, Theorem~\ref{thm:main} closes that boundary case at
length $n-1$.

\begin{lemma}
\label{lem:stabilise}
Let $\pi(u)$, $u\in\Sigma^{*}$, be linear maps of a vector space over a
field with $\pi(uv)=\pi(u)\pi(v)$ and the identity at the empty word;
let $H$ be a subspace of dimension $h$ and $W_0\subseteq H$ a subspace
of dimension $m\ge1$, and put
$W_{t+1}=W_t+\sum_{x\in\Sigma}\pi(x)W_t$, so that $W_t$ is spanned by
the vectors $\pi(u)w$ with $w\in W_0$ and $|u|\le t$.  If
$W_{h-m+1}\subseteq H$ then $W_t=W_{h-m}$ for every $t\ge h-m$; in
particular $\pi(u)W_0\subseteq H$ for every word $u$.
\end{lemma}

\begin{proof}
The chain $W_0\subseteq W_1\subseteq\dots\subseteq W_{h-m+1}\subseteq H$
has $h-m+1$ inclusions and dimensions between $m$ and $h$, so
$W_t=W_{t+1}$ for some $t\le h-m$.  Then $\pi(x)W_t\subseteq W_{t+1}=W_t$
for every letter $x$, so $W_{t'}=W_t$ for every $t'\ge t$, and
$\pi(u)W_0\subseteq W_{|u|}\subseteq W_t\subseteq H$ for every word $u$.
\end{proof}

\begin{theorem}
\label{thm:main}
Let $\mathcal A$ be a synchronizing complete automaton on $n$
states.  Then every proper nonempty $S\subseteq Q$ with $B(S)\ge0$ satisfies
$\mathrm{minext}(S)\le n-1$.
\end{theorem}

Synchronization is the only hypothesis; nothing is assumed about the
in-degree profile, connectivity, or the alphabet size.

\begin{proof}
Suppose $\mathrm{minext}(S)>n-1$.  By~\eqref{eq:defect}, every summand of
$B(S)=\sum_{t=1}^{n-1}k^{\,n-1-t}\sum_{|u|=t}(|Su^{-1}|-|S|)$ is $\le0$,
while the total is $\ge0$ by hypothesis, so every summand vanishes:
\begin{equation}
\label{eq:flat}
|Su^{-1}|=|S|\qquad\text{for all }|u|\le n-1.
\end{equation}
Put $\gamma=[S]^{\mathsf T}-\tfrac{|S|}{n}\mathbf1$, nonzero because
$0<|S|<n$; since $\pi(u)[S]^{\mathsf T}=[Su^{-1}]^{\mathsf T}$ and
$\pi(u)\mathbf1=\mathbf1$, $\mathbf1^{\mathsf T}\pi(u)\gamma=|Su^{-1}|-|S|$
for every word $u$.  So~\eqref{eq:flat} says that $\pi(u)\gamma$ lies in
the hyperplane $H$ of vectors with zero coordinate sum for every $|u|\le n-1$,
and Lemma~\ref{lem:stabilise} with $W_0=\mathrm{span}\{\gamma\}$, $m=1$
and $h=n-1$ gives $\pi(u)\gamma\in H$, that is $|Su^{-1}|=|S|$, for every
word $u$.  A reset word $r$, with $Q\cdot r=\{q_0\}$, has
$Sr^{-1}\in\{\emptyset,Q\}$, so $|Sr^{-1}|\in\{0,n\}$ differs from
$|S|$, against the previous sentence.
\end{proof}

The proof has two sources.

\begin{enumerate}
\item The chain argument, Lemma~\ref{lem:stabilise}, is Kari's Lemma 3
   \cite[\S4]{ka} in the form Steinberg gave it, with an arbitrary
   starting dimension over an arbitrary field \cite[Lemma~6]{st}; the
   proof applies it at $\dim W_0=1$ and $\dim H=n-1$, which is Kari's
   case, and that choice is why the length bound is $n-1$.
\item The
   classical argument averages $|Su^{-1}|-|S|$ over the words of length $t$
   on a balanced automaton, where the average vanishes; here the same
   average is taken on an arbitrary automaton and its value, the defect, is
   kept as the functional $B$.  The hypothesis becomes a per-subset
   inequality, $B(S)\ge0$; the boundary case $B(S)=0$ is included, which
   Section~\ref{sec:hypothesis} locates.
\end{enumerate}

\subsection{Where synchronization is used}
\label{sec:hypothesis}

Synchronization is used only after~\eqref{eq:flat}, hence only on subsets all
of whose defects vanish.  That set is contained in $\{B=0\}$ and is cut out
by the $n-1$ conditions $\sigma_t(S)=0$.  It is a subspace intersected with
the cube, which on 55.3 to 93.3 percent of
the non-Eulerian automata measured contains no proper nonempty subset.
Section~\ref{sec:qcert} narrows that locus further within $\{B\ge0\}$.

\begin{theorem}
\label{thm:krylov}
For every complete automaton,
\begin{multline*}
\{S:\ B(S)\ge0\}\ \cap\ \{S:\ \sigma_t(S)\le0\ \text{for all }t\le n-1\}\\
\ =\ \{S:\ \sigma_t(S)=0\ \text{for all }t\le n-1\},
\end{multline*}
and the right-hand side equals $\mathcal K^{\perp}\cap\{0,1\}^{Q}$ for the
Krylov space
$\mathcal K=\mathrm{span}\{w^{\mathsf T}M^{\,j}:\ 0\le j\le n-2\}$.  Since
$w^{\mathsf T}M^{\,j}\mathbf1=0$ for every $j$, $\mathcal K$ lies in the
hyperplane orthogonal to $\mathbf1$ and $\dim\mathcal K\le n-1$; when
$\dim\mathcal K=n-1$, $\mathcal K^{\perp}=\mathrm{span}(\mathbf1)$ contains
no proper nonempty $0/1$ vector.
\end{theorem}

\begin{proof}
If $B(S)\ge0$ and every $\sigma_t(S)\le0$ then
$B=\sum_tk^{\,n-1-t}\sigma_t\le0$ with strictly positive weights, so $B=0$
and every $\sigma_t=0$; the converse is immediate.  For the description,
the display in the proof of Lemma~\ref{lem:closed} gives
$\sigma_t=\sum_{j<t}k^{\,t-1-j}\rho_j$ with
$\rho_j=w^{\mathsf T}M^{\,j}[S]^{\mathsf T}$, a unitriangular change of
coordinates, so $\sigma_1=\dots=\sigma_{n-1}=0$ if and only if
$\rho_0=\dots=\rho_{n-2}=0$.
\end{proof}

Hence Theorem~\ref{thm:main} reads:  
On a subset outside
$\mathcal K^{\perp}$, either some $\sigma_t>0$, so that Corollary~\ref{cor:strict} applies,
or $B(S)<0$, where Theorem~\ref{thm:main} does not apply.  Synchronization is used
only on $\mathcal K^{\perp}\cap\{0,1\}^{Q}$.

We have empirically explored how often that set contains no proper nonempty subset.  
Theorem~\ref{thm:krylov} assumes
neither synchronization nor strong connectivity.  Over the non-Eulerian
automata of the exhaustive synchronizing strongly connected binary
populations of \ref{sec:verification} at $n=3$ to $6$ and of the
$\{d\le2\}$ strata at $n=5$ to $7$ (Table~A.1 in \ref{app:void-share}),
the share of automata with no proper nonempty subset in
$\mathcal K^{\perp}\cap\{0,1\}^{Q}$ runs from 55.3 percent ($n=6$,
$\{d\le2\}$) to 93.3 percent ($n=5$, all deviations) and is not monotone
in $n$.

\section{Antisymmetry and the Eulerian fibre}
\label{sec:antisym-fibre}

\subsection{Antisymmetry}
\label{sec:antisymmetry}

Since $B(Q)=0$, $B(Q\setminus S)=-B(S)$ at every $n$, every $k$ and every
deviation.  Hence:

\begin{corollary}
\label{cor:antisym}
On every synchronizing complete automaton, for every proper
nonempty $S$, at least one of $S$ and $Q\setminus S$ has
$\mathrm{minext}\le n-1$.  Consequently
\[\#\{S\ \text{proper nonempty}:\ \mathrm{minext}(S)>n-1\}\ \le\ 2^{\,n-1}-1.\]
\end{corollary}

\begin{proof}
$B(S)+B(Q\setminus S)=0$, so one of the two is $\ge0$; apply
Theorem~\ref{thm:main}.  The $2^{n}-2$ proper nonempty subsets fall into $2^{\,n-1}-1$
complementary pairs, each contributing at most one.
\end{proof}

Corollary~\ref{cor:antisym} applies to every proper nonempty subset, with no condition on
$B$: of a complementary pair, at most one member is stuck.  The antisymmetry holds length
by length as well, since $\mathrm{indeg}_t(Q)=k^{t}n$ gives
$\sigma_t(Q\setminus S)=-\sigma_t(S)$ for every $t$; thus Corollary~\ref{cor:antisym}
holds verbatim with the larger region $C^{\sharp}$ defined in Section~\ref{sec:csharp}
in place of $\{B\ge0\}$.  The bound $2^{\,n-1}-1$ does not drop under that
sharpening, because it comes from pairing the proper subsets into
complementary pairs and taking one from each.

\subsection{The Eulerian fibre}
\label{sec:zero-fibre}

\begin{theorem}
\label{thm:fibre}
For every complete automaton, $\beta^{\ast}=0$ if and only if
$w=0$, that is, if and only if $\mathcal A$ is Eulerian.
\end{theorem}

\begin{proof}
This is the case $c_t=k^{\,n-1-t}$ of Theorem~\ref{thm:weights} below,
where $\beta_c=\beta^{\ast}$.
\end{proof}

Call the set of Eulerian automata the Eulerian fibre; by
Theorem~\ref{thm:fibre} it is the zero set of $\beta^{\ast}$.  On an
Eulerian automaton $B\equiv0$, so on a \emph{synchronizing} Eulerian
automaton Theorem~\ref{thm:main} returns $\mathrm{minext}(S)\le n-1$ for
every proper nonempty subset, which is Kari's extension lemma
\cite[\S4, Lemma~4]{ka} in its Eulerian specialisation, for every alphabet
size and including singletons; and by Theorem~\ref{thm:fibre}, $B\equiv0$
holds only on the Eulerian automata.  Synchronization is not removable
here even though $B\equiv0$: a group automaton, one all of whose letters
act as permutations of $Q$, is Eulerian, and no proper nonempty subset of
it extends at any length (Section~\ref{sec:two-sided}).

\begin{remark}
\label{rem:deformation}
Theorems~\ref{thm:main} and~\ref{thm:fibre} together say that
Theorem~\ref{thm:main} replaces the Eulerian hypothesis of Kari's lemma
by the per-subset inequality $B(S)\ge0$, and that the inequality holds
identically exactly on the automata satisfying Kari's hypothesis.
\end{remark}

On every Eulerian fibre swept the constant $n-1$ is attained, at the
subset sizes coprime to $n$ and at no others.  Over all synchronizing
Eulerian binary automata at $n=4$ through $n=9$ (Table~A.2 gives the
sizes of the fibres), the maximum of $\mathrm{minext}$ over the subsets
of size $m$ is $n-1$ when $\gcd(m,n)=1$ and $n-2$ otherwise (Table~1).

\begin{center}
\renewcommand{\baselinestretch}{1}\footnotesize
\begin{tabular}{@{}c|cccccccc@{}}
\toprule
$n$ & $m=1$ & $2$ & $3$ & $4$ & $5$ & $6$ & $7$ & $8$ \\
\midrule
$4$ & $3$ & $\underline{2}$ & $3$ & & & & & \\
$5$ & $4$ & $4$ & $4$ & $4$ & & & & \\
$6$ & $5$ & $\underline{4}$ & $\underline{4}$ & $\underline{4}$ & $5$ & & & \\
$7$ & $6$ & $6$ & $6$ & $6$ & $6$ & $6$ & & \\
$8$ & $7$ & $\underline{6}$ & $7$ & $\underline{6}$ & $7$ & $\underline{6}$ & $7$ & \\
$9$ & $8$ & $8$ & $\underline{7}$ & $8$ & $8$ & $\underline{7}$ & $8$ & $8$ \\
\bottomrule
\end{tabular}
\end{center}

\noindent{\small\textbf{Table 1.}\ The maximum of $\mathrm{minext}$ over
the subsets of size $m$, over all synchronizing Eulerian binary automata
on $n$ states (\ref{sec:verification}).  The underlined entries are the
sizes with $\gcd(m,n)>1$, where the maximum is $n-2$; every other entry is
$n-1$.\par}
At $n=4$ and $n=6$ only the sizes $1$ and $n-1$ attain $n-1$, while at
$n=5$ and $n=7$ every size does, which reads as a split by the parity of
$n$; at $n=8$ the sizes $3$ and $5$ attain $7$, and at $n=9$ the sizes $3$
and $6$ stop at $7$.  So the constant $n-1$ of Theorem~\ref{thm:main} is
attained on the Eulerian fibre.  We can prove the attaining direction,
with an explicit family of automata.

\begin{theorem}
\label{thm:gcd}
For $n\ge3$ and $1\le k\le n-1$ let $\mathcal E(n,k)$ be the binary
automaton on $Q=\{0,\dots,n-1\}$ with $b(i)=i+1$ for $i<n-1$,
$b(n-1)=n-k$, and $a$ the identity except $a(n-k)=0$.  Then
$\mathcal E(n,k)$ is Eulerian and strongly connected, and synchronizing
if and only if $\gcd(k,n)=1$.  If $\gcd(k,n)=1$ and $k\ge3$,
exactly two proper nonempty subsets $S$ have $\mathrm{minext}(S)=n-1$;
they are complementary, and their sizes $m$ satisfy $mk\equiv\pm1\pmod n$.
Consequently, for every $n\ge4$ and every $m$ with $\gcd(m,n)=1$, some
synchronizing Eulerian binary automaton on $n$ states has a subset of
size $m$ with $\mathrm{minext}=n-1$.
\end{theorem}

The proof is in \ref{app:fibre-proofs}.  Adding letters that act as the
identity keeps an automaton Eulerian, synchronizing and strongly
connected and changes no $\mathrm{minext}$, since a word extends a subset
if and only if the word obtained by deleting its identity letters does;
so the attaining half holds over every alphabet of two or more letters.

The converse is open.  On every fibre swept, a subset size
$m$ with $\gcd(m,n)>1$ stops at $n-2$; we conjecture that this holds
at every $n$: that a subset $S$ with $\mathrm{minext}(S)=n-1$ of a
synchronizing Eulerian binary automaton has $|S|$ coprime to $n$.  The
reading of the proof in \ref{app:fibre-proofs} extends to every synchronizing Eulerian binary
automaton with a single state $x_0$ of $a$-in-degree $2$ and a single
state $y_0$ of $b$-in-degree $2$: there $\mathrm{minext}(S)=n-1$ if and only if $S$ is a
union of components of the graph on $Q$ whose edges are the pairs
$\{x_0\cdot u,y_0\cdot u\}$ with $|u|\le n-3$.  If the component partition
of that graph stopped changing at some length while it still had two or
more classes, it would be a congruence and every union of its classes
would extend at no length, against Theorem~\ref{thm:main}.  So the number
of components drops by at least one per length from $n-1$ at length $0$,
leaving the graph at length $n-3$ with at most two components.  The conjecture
asks why, when there are two, their sizes are coprime to $n$.  For the
automata just described, those with a single state of $a$-in-degree $2$,
the answer is the proposition below at $\kappa=1$.

\begin{proposition}
\label{prop:gcd-partial}
Let $\mathcal A$ be a synchronizing Eulerian automaton on $n$ states over
any alphabet, for each letter $x$ let $\kappa_x=n-|Q\cdot x|$ be its
number of collisions, and let $S$ be a proper nonempty subset with
$\mathrm{minext}(S)=n-1$.  Then every prime divisor of $\gcd(|S|,n)$ is
at most $\max_x\kappa_x$.  In particular $\gcd(|S|,n)=1$ whenever every
letter has at most one collision.  On a binary Eulerian automaton
$\kappa_a=\kappa_b$ is the number of states with two $a$-preimages,
equivalently the number with two $b$-preimages; write $\kappa$ for it.
\end{proposition}

The proof is in \ref{app:fibre-proofs}.

For prime $n$ the conjecture holds because $0<|S|<n$, so what is
open concerns composite $n$ and, by the proposition, automata with
$\kappa\ge2$; a counterexample of size $m$ would have every prime
divisor of $\gcd(m,n)$ at most $\kappa$.  The conjecture is binary: with
three letters it fails at $n=4$.  On $Q=\{0,1,2,3\}$ let the
letters $a$, $b$, $c$ send $0,1,2,3$ to $0,0,2,2$, to $0,1,3,2$ and to
$3,1,1,3$.  Every state has in-degree $3$, the automaton is synchronizing
and strongly connected, $\kappa_a=\kappa_c=2$, and $S=\{0,1\}$, with
$\gcd(|S|,n)=2$, has $\mathrm{minext}(S)=3$: every $Su^{-1}$ with
$|u|\le2$ has size $2$, and $S(cbc)^{-1}=Q$.  Over the synchronizing
strongly connected Eulerian ternary automata on four states, the first
letter up to conjugacy and the other two ordered (24,606 automata), 832
pairs of an automaton and a subset of size $2$ attain $n-1$
(\ref{sec:verification}).

Off the
fibre a singleton can be stuck, so Theorem~\ref{thm:main} is not vacuous
at singletons.

\subsection{The fibre is independent of the weighting}
\label{sec:weighting}

The weights $k^{\,n-1-t}$ entered only through the closed form~\eqref{eq:closed}.  The
zero set does not depend on them.

\begin{theorem}
\label{thm:weights}
Let $c_1,\dots,c_{n-1}>0$ be any strictly positive weights
and put $\beta_c^{\mathsf T}=\sum_{t=1}^{n-1}c_t\,w_t^{\mathsf T}$, where
$w_t^{\mathsf T}=\mathbf1^{\mathsf T}M^{t}-k^{t}\mathbf1^{\mathsf T}$ is the
$t$-step deviation, so that $\sigma_t(S)=w_t^{\mathsf T}[S]^{\mathsf T}$.
Then for every complete automaton, $\beta_c=0$ if and only if $\mathcal A$
is Eulerian.
\end{theorem}

\begin{proof}
By the display in the proof of Lemma~\ref{lem:closed},
$w_t^{\mathsf T}=w^{\mathsf T}\sum_{j<t}k^{\,t-1-j}M^{\,j}$, so
$\beta_c^{\mathsf T}=w^{\mathsf T}q_c(M)$ with
$q_c(z)=\sum_{j=0}^{n-2}\gamma_jz^{j}$ and
$\gamma_j=\sum_{t=j+1}^{n-1}c_tk^{\,t-1-j}$, hence
\[\gamma_j=c_{j+1}+k\,\gamma_{j+1}\qquad(\gamma_{n-1}:=0).\]
Put $g_j=\gamma_jk^{\,j}$.  Then $g_j=g_{j+1}+k^{\,j}c_{j+1}$, so
\[g_0>g_1>\dots>g_{n-2}=c_{n-1}k^{\,n-2}>0.\]
Let $|\zeta|\le1$ and write $G(\zeta)=\sum_jg_j\zeta^{j}=q_c(k\zeta)$.  Then
\[(1-\zeta)G(\zeta)=g_0-\sum_{j=1}^{n-2}(g_{j-1}-g_j)\zeta^{j}-g_{n-2}\zeta^{\,n-1},\]
and the modulus of the subtracted part is at most
$\sum_{j=1}^{n-2}(g_{j-1}-g_j)|\zeta|^{j}+g_{n-2}|\zeta|^{\,n-1}
\le\sum_{j=1}^{n-2}(g_{j-1}-g_j)+g_{n-2}=g_0$, every coefficient being
strictly positive.  The second inequality is strict unless $|\zeta|=1$, and
on $|\zeta|=1$ the first is strict unless
$\zeta,\zeta^{2},\dots,\zeta^{\,n-1}$ all have the same argument, which for
$n\ge3$ forces $\zeta=1$ on comparing $\zeta$ with $\zeta^{2}$; and
$G(1)=\sum_jg_j>0$.  Hence $G(\zeta)\ne0$ for $|\zeta|\le1$: every root of
$q_c$ has modulus strictly greater than $k$.  (At $n=2$ there is nothing to
prove: $q_c=\gamma_0=c_1>0$ has no roots.)  Since $M$ is non-negative
with every row sum equal to $k$, $\|M\|_\infty=k$ and every eigenvalue
of $M$ has modulus at most $k$; so
$\det q_c(M)=\prod_iq_c(\lambda_i)\ne0$, the product running over the
eigenvalues with algebraic multiplicity (Schur triangularisation; $M$
need not be diagonalisable), $q_c(M)$ is invertible, and
$w^{\mathsf T}q_c(M)=0$ forces $w=0$.  The converse direction is immediate
from $w_t^{\mathsf T}=w^{\mathsf T}\sum_{j<t}k^{\,t-1-j}M^{\,j}$.
\end{proof}

The decreasing-coefficient step is the Enestr\"om--Kakeya root-location
theorem in its strict, reversed form; see \cite{asv} for the theorem and
its sharpness.
At the weights $c_t=k^{\,n-1-t}$ of Theorem~\ref{thm:fibre},
$g_j=(n-1-j)k^{\,n-2}$ and $q_c$ is the polynomial
$p(z)=\sum_{j=0}^{n-2}(n-1-j)k^{\,n-2-j}z^{j}$ of Table~A.5.

Theorem~\ref{thm:weights} shows that Theorem~\ref{thm:fibre} does not
depend on the weights: every strictly positive weighting vanishes
identically exactly on the Eulerian automata.  By
Proposition~\ref{prop:region} below every strictly positive weighting
also certifies within the same length $n-1$; the weights $k^{\,n-1-t}$
are chosen for the closed form~\eqref{eq:closed} only.  
The regions the
weightings certify are not the same; their union over the positive cone is the
$C^{\sharp}$ of Section~\ref{sec:csharp}, which that section finds strictly
larger than $\{B\ge0\}$ at every size it reports except $n=3$.

\subsection{A single length can vanish off the Eulerian fibre}
\label{sec:single-length}

The single-length version of Theorem~\ref{thm:weights} fails.  By the display
in the proof of Lemma~\ref{lem:closed}, $w_t^{\mathsf T}=w^{\mathsf T}h_t(M)$ with
\[h_t(z)=\sum_{j<t}k^{\,t-1-j}z^{\,j}=\frac{z^{t}-k^{t}}{z-k},\]
whose roots are exactly $k\omega$ for the $t$-th roots of unity
$\omega\ne1$, all of modulus exactly $k$.  So $w_t=0$ with $w\ne0$ is
possible.  Since the roots of $h_t$
are simple, $w^{\mathsf T}h_t(M)=0$ holds if and only if $w^{\mathsf T}$ lies
in the sum of the left eigenspaces of $M$ for those of its eigenvalues of the
form $k\omega$ with $\omega^{t}=1$ and $\omega\ne1$.  Such an eigenvalue can
be carried by a proper sink component of the underlying digraph, so $M$ itself
need not be irreducible; and when $\omega$ is not real, the real vector $w$
lies in the plane spanned by a conjugate pair of left eigenvectors.

\begin{example}
\label{ex:imprimitive}
On $Q=\{0,1,2\}$ take the letters $a=(1,0,0)$ and
$b=(2,0,0)$, meaning $0\cdot a=1$, $1\cdot a=0$, $2\cdot a=0$, $0\cdot b=2$,
$1\cdot b=0$, $2\cdot b=0$.  Then $\mathrm{indeg}=(4,1,1)$ and
$w=(2,-1,-1)\ne0$, while $w^{\mathsf T}M=-2\,w^{\mathsf T}$, so
$w_2^{\mathsf T}=w^{\mathsf T}M+kw^{\mathsf T}=0$.  The defect $\sigma_2$
vanishes identically on a non-Eulerian automaton.  (This automaton is
strongly connected and not synchronizing.)  Figure~\ref{fig:bipartite}
shows the bipartition responsible.
\end{example}

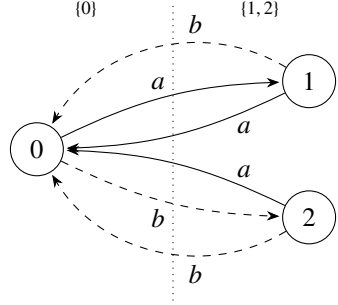
\begin{figure}[t]
\centering
\begin{tikzpicture}[>=Stealth,shorten >=1pt,
  state/.style={circle,draw,minimum size=7mm,inner sep=0pt}]
  \node[state] (q0) at (0,0.9) {$0$};
  \node[state] (q1) at (3.6,1.8) {$1$};
  \node[state] (q2) at (3.6,0) {$2$};
  \draw[dotted] (1.8,-1.1) -- (1.8,2.9);
  \node[font=\scriptsize] at (0.65,2.75) {$\{0\}$};
  \node[font=\scriptsize] at (2.95,2.75) {$\{1,2\}$};
  \draw[->] (q0) to[bend left=12] node[above,pos=.45] {$a$} (q1);
  \draw[->] (q1) to[bend left=12] node[below,pos=.18] {$a$} (q0);
  \draw[->,dashed] (q1) to[bend right=45] node[above,pos=.35] {$b$} (q0);
  \draw[->,dashed] (q0) to[bend right=12] node[below,pos=.45] {$b$} (q2);
  \draw[->] (q2) to[bend right=12] node[above,pos=.18] {$a$} (q0);
  \draw[->,dashed] (q2) to[bend left=45] node[below,pos=.35] {$b$} (q0);
\end{tikzpicture}
\caption{The witness of Example~\ref{ex:imprimitive}.  Every edge crosses
the bipartition $\{0\}\mid\{1,2\}$ (dotted), so, the automaton being
strongly connected, $M$ is imprimitive:
$w^{\mathsf T}=(2,-1,-1)$ is a left eigenvector of $M$ with eigenvalue
$-2$, of modulus exactly $k$, and $\sigma_2$ vanishes identically although
the automaton is not Eulerian.}
\label{fig:bipartite}
\end{figure}

Counting all transition tables, with no
quotient: 6 of the 639 non-Eulerian binary automata at $n=3$ and 48 of the
63,016 at $n=4$ have $\sigma_t\equiv0$ for some single $t\ge2$, against 0
of the 18,003 non-Eulerian ternary automata at $n=3$.  A single $\sigma_t$ can vanish identically on a non-Eulerian
automaton, but only when $M$ has an eigenvalue $k\omega$ with $\omega^{t}=1$,
$\omega\ne1$, because the roots of
$h_t$ lie on $|z|=k$; a strictly positive combination never does, because the
roots of $q_c$ lie strictly outside.  The strict decrease $g_j>g_{j+1}$ in
the proof of Theorem~\ref{thm:weights} is the difference between the two, obtained
from the positivity $c_t>0$ alone.  The family certifies every subset
any single length does (Proposition~\ref{prop:region}), strictly more at
the sizes measured (Section~\ref{sec:grading}); and only a strictly
positive combination vanishes identically on the Eulerian automata alone.

\section{Sharpness}
\label{sec:sharpness}

\subsection{Why the hypothesis is necessary}
\label{sec:two-sided}

Of the 146,875 automata of the unfiltered $n=5$ binary enumeration in the
quotient convention, 31,895 are not synchronizing, 82,392 are synchronizing
but not strongly connected, and the remaining 32,588 are the synchronizing
strongly connected population used elsewhere in this section.  The strict
half, Corollary~\ref{cor:strict}, returns no violation on a sweep of the whole 146,875;
the boundary case $B=0$ over the 82,392 synchronizing automata that are not
strongly connected returns no violation.  Synchronization
cannot be dropped, as one witness shows.  Let $\mathcal A$ be an automaton all of whose letters act
as permutations of $Q$.  It is Eulerian, so $\beta^{\ast}=0$ by
Theorem~\ref{thm:fibre} and $B(S)=0\ge0$ for every $S$; and every
$|Su^{-1}|=|S|$, so no proper nonempty subset extends at any length.
Without synchronization, Theorem~\ref{thm:main} would therefore fail on
every proper nonempty subset of every group automaton, on which
$B\equiv0$.  Sweeping the boundary case
over the 31,895 non-synchronizing automata returns 58,060 violations,
counted in proper-subset instances: an instance is a pair of an
automaton and a proper nonempty subset $S$, and it is a violation when
$B(S)=0$ and $\mathrm{minext}(S)>n-1$, so an automaton contributes one
violation for each such subset.

\subsection{Tightness at one integer unit}
\label{sec:tightness}

Since $B$ is integer-valued, the relaxation of Theorem~\ref{thm:main} by
one unit, to $B(S)\ge-1$, is false.  In Example~\ref{ex:cerny}, $S=\{0,1\}$ has
$B(S)=-1$ and $\mathrm{minext}(S)=3>n-1$: none of the six nonempty words of
length at most $2$ leaves $|Su^{-1}|$ above $2$, five of them leaving it at
$2$ and $ba$ at $1$, and $u=baa$ attains $3$.  The value $-1$ is attained
on every population measured: on the binary populations at $n=3$, $4$ and $5$ and
the ternary population at $n=4$ (\ref{sec:verification}), the
largest value of $B$ on a stuck subset is $-1$.  The inequality $B\le-1$ on
stuck subsets is proved, since on a synchronizing automaton every stuck subset
has all $\sigma_t\le0$ and $\sigma\not\equiv0$ (Section~\ref{sec:csharp});
what the sweeps add is that $-1$ is attained.  Stratified by deviation, that
largest value is not monotone in $d$: on the ternary $n=4$ population it
runs $-1$, $-2$, $-6$, $-4$, $-50$, $-68$ at $d=2,4,6,8,10,12$.

\subsection{Non-vacuity and the constant}
\label{sec:non-vacuity}

The set $\{B<0\}$ holds 138 of the 354 proper-subset instances at
$n=3$, 7,724 of 17,360 at $n=4$, 463,047 of 977,640 at $n=5$ and
33,481,397 of 69,596,798 at $n=6$, that is 39.0, 44.5, 47.4 and 48.1
percent, rising towards the bound $1/2$ that $B(Q\setminus S)=-B(S)$
imposes.  The gap to $1/2$ is half the share of $\{B=0\}$, on each
non-Eulerian automaton at most $1/(2\sqrt N)$ for $N$ the number of states with
$\beta^{\ast}_q\ne0$ (Proposition~\ref{prop:half-share}).  The
certified region is never empty: $B(Q\setminus S)=-B(S)$
puts one member of every complementary pair in $\{B\ge0\}$, so
Theorem~\ref{thm:main} is vacuous on no synchronizing automaton.  The
constant $n-1$ cannot be improved: in Example~\ref{ex:cerny} the subset
$\{2\}$ has $B=1\ge0$ and $\mathrm{minext}=2=n-1$.  The bound is attained
inside $\{B\ge0\}$ at every size measured (\ref{sec:verification}).

Nothing above uses a degree hypothesis.  A degree hypothesis bounds
the per-step constant of the extension method, and the family of
Kisielewicz and Szyku\l{}a \cite{ks} shows that no constant per-step
bound holds on every automaton; the family and the measurements are in
\ref{sec:bounded-deviation}.

\subsection{Grading of the two halves}
\label{sec:grading}

Both ends of the family $B_1,\dots,B_{n-1}$ certify subsets that the other
members do not.  Over the exhaustive synchronizing strongly connected $n=5$ population
(32,588 automata; 977,640 proper-subset instances), $\sigma_1(S)>0$ holds
on 314,393 instances and a further 12,064 instances have $\sigma_t(S)>0$
only at $t=n-1$, so a truncated family would miss those.  The union of the
strict half-spaces $\{\sigma_t>0\}$, $t\le T$, is strictly larger than
$\{B_T>0\}$, by 12,644 instances at $T=2$
and 29,396 at $T=4$; against the non-negativity test $B_T\ge0$ neither
region contains the other, the union holding 9,108 instances at $T=2$ and
26,826 at $T=4$ that $\{B_T\ge0\}$ misses; and
Proposition~\ref{prop:region} identifies the union over all strictly
positive reweightings as $C^{\sharp}$.

Each length $t$ is attained.  At $n=5$, for each $t\le n-1$ some subset has
$\sigma_t(S)>0$ and $\mathrm{minext}(S)=t$, on 314,393, 74,052, 9,338 and
808 instances respectively, with a witness at every size.  Fix
$n\ge2$ and $t\le n-1$, let the letter $a$ send $i$ to $i+1$ for $i<t$, send
$t$ to $1$ and fix every state above $t$, let every other letter act as the
identity, and take $S=\{t\}$.  Then $|Su^{-1}|=1$ for every word $u$ of
length less than $t$, while $S(a^{t})^{-1}=\{0,t\}$, so
$\mathrm{minext}(S)=t$; and $\mathrm{indeg}_t(S)=k^{t}+1$, so
$\sigma_t(S)=1>0$.  Hence the conclusion $\mathrm{minext}(S)\le t$ of
Corollary~\ref{cor:strict} cannot be lowered at any $t$, $n$ or $k$; the
witness is not synchronizing for $n\ge3$, but the counts above show
attainment at every $t$ on synchronizing strongly connected automata as well.

The boundary half is not graded, neither by this proof
(Theorem~\ref{thm:main} needs the full ascending chain, hence the full
length $n-1$) nor in the counts: subsets with $\sigma_t(S)=0$ for every $t\le n-2$
and $\mathrm{minext}(S)>n-2$ number 402, 2,428 and 8,528 on the exhaustive
populations at $n=4$, $5$ and $6$.

\subsection{The certified region \texorpdfstring{$C^{\sharp}$}{C-sharp}}
\label{sec:csharp}

Write
\[\begin{gathered}
C^{\sharp}=\{S:\ \exists\,t\le n-1,\ \sigma_t(S)>0\}
\cup\{S:\ \sigma_t(S)=0\ \ \forall t\le n-1\},\\
H^{\sharp}=\{S:\ \sigma_t(S)\le0\ \forall t\le n-1,\
\sigma_t(S)\ne0\ \text{for some }t\le n-1\}.
\end{gathered}\]
On a synchronizing complete automaton, every proper nonempty $S\in C^{\sharp}$
satisfies $\mathrm{minext}(S)\le n-1$: on the first clause by
Corollary~\ref{cor:strict}, on the second because $\sigma\equiv0$ gives
$B(S)=0$ and Theorem~\ref{thm:main} applies.  By Theorem~\ref{thm:krylov},
$\{B\ge0\}\subseteq C^{\sharp}$.  The enlargement is the union over all
strictly positive reweightings.

\begin{proposition}
\label{prop:region}
For every strictly positive weight vector
$c\in\mathbb R^{n-1}_{>0}$, the functional $B_c=\sum_tc_t\sigma_t$ satisfies
$B_c(S)\ge0\Rightarrow\mathrm{minext}(S)\le n-1$ for every proper nonempty
$S$, on every synchronizing complete automaton, and
\[\bigcup_{c>0}\{S:\ B_c(S)\ge0\}\;=\;C^{\sharp}.\]
\end{proposition}

\begin{proof}
Validity as for $B$: if $B_c(S)\ge0$ and no $\sigma_t>0$ then all
$\sigma_t$ vanish, and Theorem~\ref{thm:main} applies to $B$.  For the union: if some
$\sigma_t(S)>0$, choose $c_t$ large and the rest small; if every
$\sigma_t(S)\le0$ with at least one strictly negative, then $B_c(S)<0$ for
every strictly positive $c$; and if $\sigma\equiv0$ then $B_c(S)=0$ for
every $c$.
\end{proof}

In the notation of Theorem~\ref{thm:weights}, $B_c(S)=\beta_c^{\mathsf T}[S]^{\mathsf T}$,
and $B=B_c$ at $c_t=k^{\,n-1-t}$.

Passing from $\{B\ge0\}$ to
$C^{\sharp}$ shrinks $\{B<0\}$ by 3.4, 5.8 and 9.1 percent at
$n=4,5,6$ over all deviations, counted in proper-subset instances in the
quotient convention of \ref{sec:conventions}; the full series is the
first row of Table~D.1 in \ref{app:qcert}.  At $n=3$ the gain is 0,
$C^{\sharp}=\{B\ge0\}$ on all 59 automata; at
ternary $n=4$ the gain is 3.06 percent against 3.42 at $k=2$, so
$C^{\sharp}\setminus\{B\ge0\}$ stays below 4 percent of $\{B<0\}$ at
$k=3$ as at $k=2$.  The
second-moment test of Section~\ref{sec:qcert}, which is not in this linear
family, shrinks $\{B<0\}$ by 60 to 95 percent (Table~D.1; Section~\ref{sec:asymptotic} for its share as $n$ grows).  No enlargement, inside or outside the family, removes a
stuck subset: on a synchronizing automaton every stuck subset has all
$\sigma_t\le0$ and $\sigma\not\equiv0$, by Corollary~\ref{cor:strict} and
Theorem~\ref{thm:main}, hence lies in $H^{\sharp}$.

\subsection{A second-moment certificate}
\label{sec:qcert}

Every functional of Section~\ref{sec:csharp} is linear in $[S]$.  The
second moment of the deviations $|Su^{-1}|-|S|$ gives a certificate
outside that family.  Write
$x_u=|Su^{-1}|-|S|$, so that $\sigma_t(S)=\sum_{|u|=t}x_u$
by~\eqref{eq:defect}.

\begin{proposition}[Q-CERT]
\label{prop:qcert}
For every complete automaton, every proper nonempty $S\subseteq Q$ and
every $t\ge1$,
\[\sum_{|u|=t}x_u^{2}\ >\ |S|\bigl(-\sigma_t(S)\bigr)
\qquad\Longrightarrow\qquad \mathrm{minext}(S)\le t.\]
\end{proposition}

\begin{proof}
If every $x_u\le0$ then $|x_u|\le|S|$, since $|Su^{-1}|\ge0$ gives
$x_u\ge-|S|$, so
$\sum_{|u|=t}x_u^{2}\le|S|\sum_{|u|=t}|x_u|=|S|\bigl(-\sigma_t(S)\bigr)$ by
\eqref{eq:defect}; the statement is the contrapositive.
\end{proof}

Like Corollary~\ref{cor:strict}, this assumes no synchronization, no
connectivity and no condition on the degrees or on the alphabet; and it
contains the corollary, since when $\sigma_t(S)>0$ the right-hand side is
negative and the hypothesis holds automatically.  The threshold is not
sharp, though the sharp one uses the same three quantities.  Write $m=-\sigma_t(S)$, $c=|S|$,
$q=\lfloor m/c\rfloor$ and $r=m-qc$.

\begin{proposition}[Q-CERT+]
\label{prop:qcertplus}
Under the hypotheses of Proposition~\ref{prop:qcert},
\[\sum_{|u|=t}x_u^{2}\ >\ q\,|S|^{2}+r^{2}
\qquad\Longrightarrow\qquad \mathrm{minext}(S)\le t,\]
and $qc^{2}+r^{2}\le mc$, with equality only when $r=0$.
\end{proposition}

\begin{proof}
If every $x_u\le0$ then $(|x_u|)_u$ lies in
$\{0\le y_u\le c,\ \sum_uy_u=m\}$, over which the convex function
$\sum_uy_u^{2}$ is maximised at a vertex, that is at $q$ coordinates equal
to $c$, one equal to $r$ and the rest $0$, with value $qc^{2}+r^{2}$; and
$qc^{2}+r^{2}\le qc^{2}+rc=mc$, with equality only when $r=0$.
\end{proof}

(Q-CERT+) is the sharpest threshold in $\sum_{|u|=t}x_u^{2}$, but the
attainable second moments are sparse, so a threshold is not the sharpest
test on the same data; the finer membership test on the same data, its
cost, and the measured shares are in \ref{app:qcert}.  In brief, the
second moment is a pair count on the automaton on $Q\times Q$, costing
$O(kn^{3})$ for all lengths $t\le n-1$ against $O(kn^{2})$ for
$\beta^{\ast}$.  Over the eleven populations of \ref{sec:verification},
in proper-subset instances and in the conventions stated in Table~D.1,
it removes between 60.14 and 95.24 percent of $\{B<0\}$ under
(Q-CERT) and between 60.14 and 95.51 percent under (Q-CERT+),
against at most 9.05 percent for $C^{\sharp}$ (Table~D.1).  So the second
moment is not a refinement of $C^{\sharp}$; on the exhaustive $n=4$
populations the gap is wider at $k=3$ than at $k=2$.

The test removes no stuck subset, nor can it: a subset
with $\mathrm{minext}(S)>n-1$ has every $x_u\le0$ at every $t\le n-1$, so
neither hypothesis holds (measured as 0 stuck subsets certified on all
eleven populations of \ref{sec:verification}, as an implementation
check).  And it does not lower the length.  For each of the strict test
$\sigma_t(S)>0$, (Q-CERT) and (Q-CERT+), the least length at which its
hypothesis holds, maximised over the subsets of a given size that it certifies, is
$n-1$ at every size $1,\dots,n-1$; so is the largest $\mathrm{minext}$ over
the subsets of that size that extend within $n-1$.  This holds on the same eleven
populations at $n=3$ to $7$, $k=2$ and $k=3$, none an Eulerian fibre.  For
no test and no size is that maximum below $n-1$, in agreement with
Section~\ref{sec:no-reset}.

Let $\Phi$ be any functional of
the multiset $\{x_u\}_{|u|=t}$ and of $|S|$ (the latter is not determined
by the former) such that $\Phi>0$ implies $\max_{|u|=t}x_u>0$; then
\[\{S:\Phi(S)>0\}\ \subseteq\ \{S:\max_{|u|=t}x_u>0\}\ \subseteq\
\{S:\mathrm{minext}(S)\le t\},\]
so a subset with $\mathrm{minext}(S)>t$ has every $x_u\le0$ and is certified
by no such $\Phi$: every functional of this shape certifies at most what
the maximum does.  Both forms above are of that shape; the second
moment alone is not ($\sum_{|u|=t}x_u^{2}=2>0$ on the multiset
$-1,-1,0,0$); it is the term $|S|\sigma_t(S)$ that supplies the
property.  The maximum is not a usable certificate.  It costs
$\Theta(k^{t})$ preimage computations against $O(ktn)$ for $\sigma_t$, and
deciding $\mathrm{minext}(S)\le\ell$ with $\ell$ part of the input is
NP-complete at $|S|=n-1$ on strongly connected synchronizing binary
automata \cite[Thm.~13]{bfs}; the reduction produces $\ell=2m+3$ on
at least $3(m+1)$ states, so $\ell\le n-1$ at every instance, within the
length used here.  A polynomial-time evaluation
of $\max_{|u|=t}x_u$ would decide such an instance in at most $n-1$ calls.

The second inclusion of the display is strict, because extending in exactly
$t$ steps and extending within $t$ steps are different conditions; a witness,
Lemma~\ref{lem:monotone}, locating the gap on $\{d\ge2k\}$, and the gap
census are in
\ref{app:qcert}.  Theorem~\ref{thm:main} is not subject to this
inclusion, because it is not of this shape: it shows a multiset value to be
unrealisable under synchronization.

On the locus where every
$\sigma_t(S)=0$, which contains every subset on which synchronization is
used (Section~\ref{sec:hypothesis}), the hypothesis of (Q-CERT) holds if
and only if the $x_u$ are not all zero.  So within $\{B\ge0\}$ synchronization is
used only on the flat locus
\[\{S:\ |Su^{-1}|=|S|\ \text{for all }|u|\le n-1\},\]
which is~\eqref{eq:flat} in the proof of Theorem~\ref{thm:main} and lies
inside $\mathcal K^{\perp}\cap\{0,1\}^{Q}$.  On a
synchronizing automaton that locus contains no proper nonempty subset:
flatness gives every $\sigma_t(S)=0$, hence $B(S)=0\ge0$, hence
$\mathrm{minext}(S)\le n-1$ by Theorem~\ref{thm:main}, against flatness.
Over the non-synchronizing automata it is nonempty: the 58,060
boundary violations of Section~\ref{sec:two-sided} are its proper nonempty
members, since $B(S)=0$ with $\mathrm{minext}(S)>n-1$ forces every
$\sigma_t(S)=0$ and then every $|Su^{-1}|=|S|$.

\subsection{The shares as the state set grows}
\label{sec:asymptotic}

The shares of Table~D.1 rise with $n$ within each series;
Theorem~\ref{thm:share} and Proposition~\ref{prop:half-share} bound them
on every automaton.  For a word
$u$ let $\mu(u)$ be the number of unordered pairs of distinct states
$p,p'$ with $p\cdot u=p'\cdot u$, and for $t\ge1$ put
\[
P_t=k^{-t}\sum_{|u|=t}\mu(u),\qquad p_t=P_t\Big/\binom n2,\qquad
\Delta_t=\max_q\mathrm{indeg}_t(q)/k^{t},
\]
the mean number and the mean fraction of pairs merged by a word of
length $t$, and the largest $t$-step in-degree relative to its mean
$k^{t}$, so that $\Delta_t\ge1$ and $\Delta_t/n$ is the largest fraction
of the $nk^{t}$ landings $(p,u)$, $|u|=t$, received by one state.  In
the notation of \ref{app:qcert},
$P_t=\bigl(k^{-t}\sum_q\mathrm{pairdeg}_t(q,q)-n\bigr)/2$, so both are
computable in polynomial time.

\begin{theorem}
\label{thm:share}
Let $\mathcal A$ be a complete automaton and $t\ge1$ with $P_t>0$.  The hypothesis
of (Q-CERT) at length $t$ fails on at most
\[
\Bigl(\frac{8n(\Delta_t+1)}{P_t}+\frac{4n^{3}(\Delta_t-1)}{P_t^{2}}\Bigr)2^{n}
\quad\text{and on at most}\quad
3\exp\Bigl(-\frac{P_t^{2}}{32\,n^{3}(\Delta_t+1)^{2}}\Bigr)2^{n}
\]
of the $2^{n}$ subsets of $Q$.  If $\mathcal A$ is not Eulerian and
$t\le n-1$, the share of the subsets with $B(S)<0$ on which (Q-CERT) at
length $t$ fails to establish $\mathrm{minext}(S)\le n-1$ is at most
four times either bound.
\end{theorem}

In terms of $p_t$ the first bound is
$16(\Delta_t+1)/(p_t(n-1))+16n(\Delta_t-1)/(p_t^{2}(n-1)^{2})$ and the
second has exponent $p_t^{2}(n-1)^{2}/(128\,n(\Delta_t+1)^{2})$.  So
along any sequence of non-Eulerian automata on which, for some $t\le n-1$, the
merged fraction $p_t$ stays bounded below and $\Delta_t=o(n)$, the share
of $\{B<0\}$ that the second moment certifies within $n-1$ tends to
$1$, exponentially fast in $n$ when $\Delta_t$ stays bounded.  On a
group automaton no word merges a pair, so $P_t=0$; there every $x_u$
vanishes and the hypothesis of (Q-CERT) holds nowhere.  A letter that
sends every state to one state $q_0$ makes $\Delta_t\ge n/k$, as the
$k^{t-1}$ words ending with it land at $q_0$; there the first bound
exceeds $16/k$ and the exponent of the second is below $k^{2}/(128n)$,
so at $k\le16$ neither bound drops below $1$.  On \v{C}ern\'y's automata
$C_n$, where as in Example~\ref{ex:cerny} the letter $a$ is the cycle
$0\to1\to\dots\to n-1\to0$ and $b$ sends $0$ to $1$ and fixes every
other state, the letter $b$ alone merges, joining the states $0$ and
$1$ into a fibre at $1$, after which no state maps to $0$ until a
letter $a$ does.  A merged fibre moves one state forward at each $a$
and, away from $0$, is fixed by $b$.  A later $b$ merges it again only
if it is at $0$, after $n-1$ letters $a$, or at $1$ with a state at
$0$, which needs an $a$ to map a state to $0$ and then $n$ letters $a$
in all to return the fibre to $1$.  Within a word of length $n-1$ every
merge therefore joins two singleton fibres and adds one pair, so a word
merges at most $n-1$ pairs and $p_{n-1}\le2/n$: along \v{C}ern\'y's
automata the merged fraction $p_{n-1}$ tends to $0$.

By the proof below, the hypothesis of (Q-CERT) at length $t$ fails on
$S$ only if the normalised defect $\sigma_t(S)/k^{t}$
is at most $-p_t(n-1)/8$, or if the second moment $\sum_{|u|=t}x_u^{2}$
is at most $k^{t}P_t/4$, half its mean over the subsets.  When $P_t>0$
the subsets of the first kind are a share at most
$4n^{3}(\Delta_t-1)/P_t^{2}$ and at most
$2\exp(-P_t^{2}/(32n^{3}(\Delta_t+1)^{2}))$; those of the second are a
share at most $8n(\Delta_t+1)/P_t$ and at most
$\exp(-P_t^{2}/(32n^{3}(\Delta_t+1)^{2}))$.  So the test certifies every
subset whose normalised defect exceeds $-p_t(n-1)/8$ and whose second
moment exceeds $k^{t}P_t/4$.  For $t\le n-1$ the failing share is at
least $2^{1-n}$ plus the share of the stuck subsets, since $\emptyset$,
$Q$ and every stuck subset fail; it therefore lies between $2^{1-n}$ and
$3\exp(-p_t^{2}(n-1)^{2}/(128\,n(\Delta_t+1)^{2}))$.  Both bounds are
exponentially small in $n$ when $p_t$ stays bounded below and $\Delta_t$
stays bounded.

\begin{proof}
Let $\xi_q$, $q\in Q$, be independent with
$\mathbb P(\xi_q=1)=\mathbb P(\xi_q=0)=1/2$ and $S=\{q:\xi_q=1\}$, so
that $S$ is uniform among the $2^{n}$ subsets, and put
$\eta_q=\xi_q-1/2$.  For a word $u$ write $c_{u,q}=|qu^{-1}|-1$; then
$\sum_qc_{u,q}=0$, $x_u=\sum_q\xi_qc_{u,q}=\sum_q\eta_qc_{u,q}$, and
$\sum_qc_{u,q}^{2}=\sum_q|qu^{-1}|^{2}-n=2\mu(u)$, since the ordered
pairs $(p,p')$ with $p\cdot u=p'\cdot u$ number $n+2\mu(u)$.  Let
$F=\sum_{|u|=t}x_u^{2}$ and $G=\sigma_t(S)=\sum_q\xi_qe_q$ with
$e_q=\mathrm{indeg}_t(q)-k^{t}$, so that $\sum_qe_q=0$.  The
expectations are $\mathbb E\,x_u=0$, $\mathbb E\,x_u^{2}=\mu(u)/2$,
$\mathbb E\,F=k^{t}P_t/2$ and $\mathbb E\,G=0$.  The hypothesis of
(Q-CERT) fails only if $F\le|S|(-G)\le n|G|$, hence only if
$F\le\mathbb EF/2$ or $|G|\ge\mathbb EF/(2n)$.

For $G$, $\mathrm{Var}\,G=\sum_qe_q^{2}/4$, and with
$a_q=\mathrm{indeg}_t(q)/k^{t}$, which satisfy $a_q\ge0$, $\sum_qa_q=n$
and $\max_qa_q=\Delta_t$,
$\sum_qe_q^{2}=k^{2t}\bigl(\sum_qa_q^{2}-n\bigr)\le k^{2t}n(\Delta_t-1)$.
Chebyshev's inequality gives
$\mathbb P(|G|\ge\mathbb EF/(2n))\le4n^{3}(\Delta_t-1)/P_t^{2}$, and
Hoeffding's inequality \cite{ho},
$\mathbb P(|G|\ge\lambda)\le2\exp(-2\lambda^{2}/\sum_qe_q^{2})$, gives
$\mathbb P(|G|\ge\mathbb EF/(2n))\le2\exp(-P_t^{2}/(8n^{3}(\Delta_t-1)))$,
which is at most $2\exp(-P_t^{2}/(32n^{3}(\Delta_t+1)^{2}))$; when
$\Delta_t=1$ every $e_q$ vanishes and $G=0$.

For $F$, let $C$ be the $k^{t}\times n$ matrix $(c_{u,q})$ and
$\Xi=C^{\mathsf T}C$, so that
$F=\sum_{|u|=t}\bigl(\sum_q\eta_qc_{u,q}\bigr)^{2}
=\sum_{q,q'}\eta_q\eta_{q'}\Xi_{q,q'}
=\tfrac14\mathrm{tr}\,\Xi+2\sum_{q<q'}\eta_q\eta_{q'}\Xi_{q,q'}$,
using $\eta_q^{2}=1/4$.  The products $\eta_q\eta_{q'}$ over the
unordered pairs have mean $0$, variance $1/16$, and are pairwise
uncorrelated, since a product of four factors $\eta_q$ in which some
index occurs once has mean $0$; hence $\mathbb EF=\mathrm{tr}\,\Xi/4$ and
\[
\mathrm{Var}\,F=\tfrac14\sum_{q<q'}\Xi_{q,q'}^{2}
\le\tfrac18\|\Xi\|_F^{2}\le\tfrac18\|\Xi\|\,\mathrm{tr}\,\Xi,
\]
where $\|\Xi\|_F$ is the Frobenius norm and $\|\Xi\|$ the largest
eigenvalue of the positive semidefinite $\Xi$.  The eigenvalue $\|\Xi\|$,
the square of the largest singular value of $C$, is at most
$\max_u\sum_q|c_{u,q}|\cdot\max_q\sum_u|c_{u,q}|
\le2n\cdot k^{t}(\Delta_t+1)$, since
$\sum_q|c_{u,q}|\le\sum_q(|qu^{-1}|+1)=2n$ and
$\sum_u|c_{u,q}|\le\mathrm{indeg}_t(q)+k^{t}$.  Chebyshev's inequality
gives
$\mathbb P(F\le\mathbb EF/2)\le4\,\mathrm{Var}F/(\mathbb EF)^{2}
\le8\|\Xi\|/\mathrm{tr}\,\Xi\le8n(\Delta_t+1)/P_t$, as
$\mathrm{tr}\,\Xi=2k^{t}P_t$.  For the exponential bound, changing one
$\xi_q$ changes every $x_u$ by $\pm c_{u,q}$ and, as $|x_u|\le n$,
changes $F$ by at most $2n\sum_u|c_{u,q}|\le2nk^{t}(\Delta_t+1)$;
McDiarmid's inequality \cite{mc} gives
$\mathbb P(F\le\mathbb EF/2)
\le\exp\bigl(-2(\mathbb EF/2)^{2}/(4n^{3}k^{2t}(\Delta_t+1)^{2})\bigr)
=\exp\bigl(-P_t^{2}/(32n^{3}(\Delta_t+1)^{2})\bigr)$.  Adding the
probabilities of the two events gives the two bounds.

For the share of $\{B<0\}$, $\beta^{\ast}\ne0$ by Theorem~\ref{thm:fibre}; on
$\{B=0\}$ the coordinates of $[S]$ other than one $q$ with
$\beta^{\ast}_q\ne0$ determine $\xi_q$, so $|\{B=0\}|\le2^{n-1}$, and
as $S\mapsto Q\setminus S$ exchanges $\{B>0\}$ and $\{B<0\}$, at least
$2^{n-2}$ subsets have $B(S)<0$.  The failing subsets among them are at
most the totals above, and where the hypothesis holds
Proposition~\ref{prop:qcert} gives $\mathrm{minext}(S)\le t\le n-1$.
\end{proof}

With $a_q=\mathrm{indeg}_t(q)/k^{t}$ as in the proof, let
$V_t=n^{-1}\sum_q(a_q-1)^{2}$, the variance of the relative $t$-step
in-degrees over $Q$, so that $\sum_qe_q^{2}=k^{2t}nV_t$ and
$V_t\le\Delta_t-1$.  The proof bounds $\sum_qe_q^{2}$ by
$k^{2t}n(\Delta_t-1)$; with the exact value, the second term of the
first bound reads $4n^{3}V_t/P_t^{2}$, so the test certifies all but a
vanishing share of the subsets when $\Delta_t=o(P_t/n)$ and
$P_t\gg n^{3/2}\sqrt{V_t}$.  Below that scale a fixed share of the
subsets fails.

\begin{proposition}
\label{prop:barrier}
Let $\mathcal A$ be a complete automaton with $n\ge100$ states and
$t\ge1$ with $V_t>0$.  If $P_t\le n^{3/2}\sqrt{V_t}/425$, the
hypothesis of (Q-CERT) at length $t$ fails on more than $2.7$ percent
of the $2^{n}$ subsets of $Q$, all of them proper and nonempty with
$\sigma_t(S)<0$.
\end{proposition}

\begin{proof}
With $\xi_q$, $\eta_q$, $c_{u,q}$, $e_q$, $F$ and $G$ as in the proof
of Theorem~\ref{thm:share}, $G=\sum_q\eta_qe_q$ since $\sum_qe_q=0$,
so $\mathbb EG^{2}=\|e\|^{2}/4$, where $\|e\|^{2}=\sum_qe_q^{2}$, and
$\mathbb EG^{4}=\frac1{16}\sum_qe_q^{4}+\frac3{16}\sum_{q\ne q'}e_q^{2}e_{q'}^{2}
\le\frac3{16}\|e\|^{4}=3(\mathbb EG^{2})^{2}$.  The Paley--Zygmund
inequality, $\mathbb P(Z>\tfrac12\mathbb EZ)\ge\tfrac14(\mathbb EZ)^{2}/\mathbb EZ^{2}$
for $Z\ge0$, applied to $Z=G^{2}$ gives
$\mathbb P(G^{2}>\|e\|^{2}/8)\ge1/12$, and as $S\mapsto Q\setminus S$
maps $G$ to $-G$, $\mathbb P(G<-\|e\|/(2\sqrt2))\ge1/24$.  As
$V_t>0$, some $\mathrm{indeg}_t(q)\ne k^{t}$, so some word of length
$t$ is not a permutation and $P_t>0$; Markov's
inequality and $\mathbb EF=k^{t}P_t/2$ give
$\mathbb P(F\ge50k^{t}P_t)\le1/100$; Hoeffding's inequality gives
$\mathbb P(|S|<n/3)\le e^{-n/18}$; and $\mathbb P(S=Q)=2^{-n}$.  For
$n\ge100$ the four events $G<-\|e\|/(2\sqrt2)$, $F<50k^{t}P_t$,
$|S|\ge n/3$ and $S\ne Q$ hold together with probability at least
$1/24-1/100-e^{-100/18}-2^{-100}>0.027$.  On that event $S$ is proper
and nonempty, $\sigma_t(S)=G<0$ and
$|S|(-\sigma_t(S))>n\|e\|/(6\sqrt2)$, while $\|e\|=k^{t}\sqrt{nV_t}$
and the hypothesis on $P_t$ give
$F<50k^{t}P_t\le\frac{50}{425}\,n\|e\|<n\|e\|/(6\sqrt2)$; so the
hypothesis of (Q-CERT) fails on $S$.
\end{proof}

So for $n\ge100$ the second moment certifies at least $97.3$ percent
of the subsets at length $t$ only where the words of length $t$ merge
more than $n^{3/2}\sqrt{V_t}/425$ pairs on average, a merged fraction
$p_t$ of order at least $\sqrt{V_t/n}$.  By
its definition, $p_{n-1}$ is the probability that a uniformly random
word of length $n-1$ merges a uniformly random pair of distinct states.
If $\Delta_{n-1}\le n^{1-\delta}$ and $p_{n-1}\ge n^{-\varepsilon}$
with $\varepsilon<\delta/2$, the first bound of Theorem~\ref{thm:share}
is $O(n^{2\varepsilon-\delta})$: along non-Eulerian automata the shares
of Table~D.1 tend to $1$ with $p_{n-1}$ as small as $n^{-\varepsilon}$.
Whether, on all but a vanishing share of the synchronizing strongly
connected automata on $n$ states, $\Delta_{n-1}\le n^{1-\delta}$ and a
uniformly random word of length $n-1$ merges a uniformly random pair
of distinct states with probability at least $n^{-\varepsilon}$, for
some $\delta>0$ and $\varepsilon<\delta/2$, is open
(Section~\ref{sec:conclusion}).

\begin{proposition}
\label{prop:half-share}
On every complete automaton the subsets with $B(S)\ge0$ number
$2^{n-1}+|\{B=0\}|/2$, and if $\beta^{\ast}$ has $N\ge1$ nonzero
entries then $|\{B=0\}|\le2^{n-N}\binom{N}{\lfloor N/2\rfloor}\le2^{n}/\sqrt N$.
The share of $\{B\ge0\}$ therefore lies between $1/2$ and
$1/2+1/(2\sqrt N)$, and equals $1/2$ outside $\{B=0\}$.
\end{proposition}

\begin{proof}
The map $S\mapsto Q\setminus S$ is a bijection of $\{B>0\}$ onto
$\{B<0\}$, since $B(Q\setminus S)=-B(S)$; so
$|\{B>0\}|=(2^{n}-|\{B=0\}|)/2$, and adding $|\{B=0\}|$ gives the
count.  For the bound, the states with $\beta^{\ast}_q=0$ contribute a
factor $2^{n-N}$, and on the $N$ states with $\beta^{\ast}_q\ne0$ the
vectors $\xi\in\{0,1\}^{N}$ on which $\sum_q\beta^{\ast}_q\xi_q$ takes a
given value correspond, through $\xi_q=(1+s_q)/2$, to the sign vectors
$s$ on which $\sum_q\beta^{\ast}_qs_q$ takes a given value, which number at most
$\binom{N}{\lfloor N/2\rfloor}$ by Erd\H{o}s's theorem on the
Littlewood--Offord problem \cite{er}.  The inequality
$\binom{2j}{j}\le4^{j}/\sqrt{2j+1}$ for $j\ge1$ holds because
$\binom{2j}{j}4^{-j}=\prod_{i\le j}\frac{2i-1}{2i}$ and
$\frac{2i-1}{2i}\le\frac{2i}{2i+1}$ give
$\bigl(\binom{2j}{j}4^{-j}\bigr)^{2}\le\prod_{i\le j}\frac{2i-1}{2i+1}=\frac1{2j+1}$;
hence $\binom{N}{\lfloor N/2\rfloor}\le2^{N}/\sqrt N$ at even $N=2j$,
at odd $N=2j+1$ because $\binom{2j+1}{j}<2\binom{2j}{j}$, and at $N=1$
directly.
\end{proof}

The certificate $B$ therefore certifies half of the subsets in the
limit as the support of $\beta^{\ast}$ grows, and the second moment all
but a vanishing share under the hypotheses of Theorem~\ref{thm:share};
the share of $C^{\sharp}$ lies between them, since
$C^{\sharp}\supseteq\{B\ge0\}$, and its limit is open.

\section{Relation to previous work}
\label{sec:previous}

For $\alpha>0$, an automaton is $\alpha$-extensible \cite{vo} if every
subset $P$ of states with
$2\le|P|<n$ admits a word $v$ with $|v|\le\alpha n$ and $|Pv^{-1}|>|P|$.  For
$n\ge3$, an $\alpha$-extensible automaton is synchronizing, with
$\mathrm{rt}\le1+\alpha n(n-2)$ \cite[Prop.~1]{vo}.  That reference is a
survey; the statement goes back at least to Dubuc's paper on circular
automata, where a set $E$ of words is called regressive if every proper
nonempty subset of states is extended by some word of $E$, and a regressive
set is shown to yield a reset word of length at most
$1+(n-2)\max_{w\in E}|w|$ \cite[p.~22]{du}.  The method is older than either
statement: Rystsov's quadratic bound for regular automata \cite{ry} is, in
the classification of \cite{vo}, 2-extensibility run through the same
proposition, and the \v{C}ern\'y automata themselves are regular.

In that vocabulary, Theorem~\ref{thm:main} returns on $\{B\ge0\}$ an
extending word of length at most $n-1$, one below the length $n$ of
$1$-extensibility, for every proper nonempty subset of that region,
singletons included.  It is therefore a per-subset, polynomial-time
checkable sufficient condition for a subset to satisfy the $\alpha=1$
condition with one unit to spare, identically satisfied on the Eulerian automata and
only there (Theorem~\ref{thm:fibre}).  On $\{B<0\}$ it does not apply,
while $\alpha$-extensibility is a hypothesis about every subset
simultaneously, so no reset bound follows; Section~\ref{sec:no-reset} gives
the arithmetic reason.

Not every strongly connected synchronizing automaton is
$\alpha$-extensible for some constant $\alpha$: Kisielewicz and Szyku\l{}a exhibit extending words
whose shortest length is quadratic in $n$, on the family of
\ref{sec:bounded-deviation}, and conclude that the cubic bound on
the reset threshold cannot be improved in general by means of the extension
method \cite{ks}.

\subsection{Kari's paper}
\label{sec:kari}

Kari proves \v{C}ern\'y's conjecture on Eulerian automata, with the bound
$(n-1)(n-2)+1$ \cite{ka}.  Three parts of that paper are used here.

First, the ascending-chain lemma of Section~\ref{sec:main-theorem},
Lemma~\ref{lem:stabilise}, is Kari's Lemma 3 \cite[\S4]{ka} in
Steinberg's generality.  Kari's lemma states that for a subspace $U$ of
$\mathbb R^{n}$, a vector $x$ and linear maps $\varphi_a$, if some word
takes $x$ out of $U$, then some word of length at most $\dim U$ does;
Steinberg's Lemma 6 \cite{st} is the later generalisation to an
arbitrary starting dimension $m$ over any field, with the bound
$\dim U-m+1$.  The proof of Theorem~\ref{thm:main} uses the case
$m=1$, Kari's, with the same bound; the proof of
Proposition~\ref{prop:gcd-partial} uses $m=2$ over $\mathbb F_p$.

Second, Kari's extension lemma \cite[\S4, Lemma~4]{ka} is stated for a
general, not necessarily Eulerian, synchronizing automaton, and in the
weight $\langle\cdot,e\rangle$, where $e$ is the
positive integer left eigenvector of the adjacency matrix with eigenvalue
$k$; the weight goes back to Friedman \cite{fr}, as Kari notes in his \S2.
The Eulerian hypothesis enters only where $e=\mathbf1$ turns weight into
cardinality.  In particular the lemma applies to arbitrary vectors off the
line of constant vectors, hence to singletons; the distinction between
weight and cardinality is Kari's too.

Third, the one-step counting identity is in Kari, in weight form:
$\sum_{a\in\Sigma}|f_a^{-1}(x)|=k|x|$ for every vector $x$ and every
automaton, Eulerian or not, with the dichotomy that either every letter
preserves the weight or some letter strictly increases it; the two-letter
case, for subsets, is in the proof of Friedman's Theorem~2.1
\cite[p.~1134]{fr}.  This is the $t=1$ ancestor of Lemma~\ref{lem:defect}:
Kari's identity is exact because $e$ is an eigenvector.  Replacing the
eigenvector by the uniform vector destroys the invariance; the present
paper keeps the resulting defect.  Kari
extends the dichotomy to an arbitrary length, if $|f_w^{-1}(x)|\ne|x|$ for
some $w$ then $|f_u^{-1}(x)|>|x|$ for some $u$ of the same length
\cite[p.~229]{ka}, which is the origin of the inclusion of
Section~\ref{sec:qcert}; the same substitution separates the two,
since in the weight no preimage at a given length differs from the subset,
while in cardinality one may be strictly smaller.

Kari's concluding section, finally, argues that the hard case for a general
proof is the nearly but not fully balanced digraphs, that is, the automata
with $w\ne0$ and $d$ small, on which $B$ is not identically zero.

\subsection{Antecedents of the family}
\label{sec:antecedents}

Let $\mathcal A^{(t)}$ denote the $t$-th power automaton: the same state
set, the alphabet $\Sigma^{t}$ with $k^{t}$ letters, each word of length $t$
acting as a single letter.  The power automaton is complete, its in-degrees
are $\mathrm{indeg}_t$, and it is Eulerian if and only if
$\mathrm{indeg}_t\equiv k^{t}$.  The counting identity of
Lemma~\ref{lem:defect} is Kari's one-step identity applied to
$\mathcal A^{(t)}$, with $\sigma_t$ the defect of that identity there.
Corollary~\ref{cor:strict} therefore reads: if $S$ is
above average in $t$-step in-degree for some length $t\le n-1$, then
$\mathrm{minext}(S)\le t$.

The passage to power automata is Steinberg's.  His averaging lemma
\cite[Lemma~2]{st} takes a probability distribution $P_1$ on words;
evaluated at $P_1$ uniform on $\Sigma^{t}$ and $R=Q$, its hypothesis reads
$\mathrm{indeg}_t\equiv k^{t}$, that is, $\mathcal A^{(t)}$ is Eulerian, and
returns $\mathrm{rt}\le c+(n-2)(n-c+t)$.

So the identity of Lemma~\ref{lem:defect} and the strict half it yields
(Corollary~\ref{cor:strict}), the graded family $B_T$, the weights
$k^{\,n-1-t}$, the algebraic antisymmetry $B(Q\setminus S)=-B(S)$ of
Section~\ref{sec:antisymmetry} and the closed form~\eqref{eq:closed} are
each a one-line consequence of Kari's identity together
with~\eqref{eq:identity}, as the chain argument in the proof of
Theorem~\ref{thm:main} is Kari's Lemma 3.  Section~\ref{sec:novelty}
lists what is added.

On an Eulerian automaton every $\mathcal A^{(t)}$ is Eulerian and every
$\sigma_t$ vanishes identically (Section~\ref{sec:zero-fibre}), so the
family vanishes identically exactly where Kari's theorem applies; a single $\sigma_t$
can also vanish identically off the Eulerian fibre
(Example~\ref{ex:imprimitive}), which a strictly positive
combination cannot (Theorem~\ref{thm:weights}).

\subsection{The stationary-vector functional}
\label{sec:stationary}

Berlinkov \cite{be} bounds the reset threshold of a synchronizing strongly
connected automaton through its stationary distribution: a vector orthogonal to the stationary distribution
$\mu$ and a shortest word making the $\mu$-pairing positive, of length at
most $n-1$ \cite[Thm~1]{be}.  The resulting reset bound contains a factor $L$,
the least common multiple of the denominators of $\mu$ \cite[Cor.~1]{be},
and his generalisation is parametrised by the number of exceptional
states.  The extension there is in $\mu$-mass, not in cardinality, and
passing from growth in $\mu$-mass to growth in cardinality requires a degree hypothesis;
and his parameter is the number of exceptional states, where ours is a
per-subset inequality with the degree profile left arbitrary.  The bound
is a minimum over the probability vector on the alphabet,
and at the uniform choice, over the exhaustively enumerated stratum of the
binary six-state population on which the first letter is a permutation,
$L$ attains 135, at which $1+(n-1)(L-2)$ evaluates to 666, against the
exact maximum reset threshold 25 at that size.

The two certificates can also be compared directly.  For a strongly connected
automaton let $e$ be the positive left eigenvector of $M$ with eigenvalue
$k$, normalised to coprime positive integers, and attach to each state the
integer $\varepsilon_q=n\,e_q-\langle\mathbf1,e\rangle$, so that
$\sum_q\varepsilon_q=0$; write $B_e(S)=\sum_{q\in S}\varepsilon_q$.  This is
Friedman's weight \cite{fr}, centred: $B_e(S)\ge0$ says that $S$ carries at
least the average weight per state.  The functional $B_e$ shares four
properties of $B$: it is linear and integer-valued; $B_e(Q)=0$, so
the antisymmetry of Section~\ref{sec:antisymmetry} holds for it verbatim;
it is computable in polynomial time; and $\varepsilon=0$ if and only if $e$
is constant, that is, if and only if $\mathcal A$ is Eulerian, so it
degenerates onto Kari's hypothesis just as the family of
Section~\ref{sec:certificate} does.  It is moreover the limit of the family.

\begin{proposition}
\label{prop:cesaro}
For every strongly connected complete automaton and every $S\subseteq Q$,
\[\frac1T\sum_{t=1}^{T}\frac{\sigma_t(S)}{k^{t}}
\;\longrightarrow\;\frac{B_e(S)}{\langle\mathbf1,e\rangle}
\qquad(T\to\infty),\]
and when $M$ is primitive the unaveraged sequence $\sigma_t(S)/k^{t}$
converges to the same limit.
\end{proposition}

The proof is in \ref{app:stationary}.

The functional $B_e$ is not a member of the family of
Proposition~\ref{prop:region}, which is closed;
Corollary~\ref{cor:strict} says nothing about it: a Ces\`aro limit of the
normalised defects can be non-negative while every single defect is
negative.

The implication fails.

\begin{proposition}
\label{prop:be-fails}
The functional $B_e$ does not certify extension within $n-1$.  There are
synchronizing strongly connected automata with a proper subset $S$ such
that $B_e(S)\ge0$, and others with $B_e(S)>0$, while
$\mathrm{minext}(S)\ge n$; witnesses exist with repeated letters on five
states, with pairwise distinct letters on seven states and, with
repeated letters, on every number of states from five on
(Proposition~\ref{prop:family}).
\end{proposition}

The witnesses, four automata each checked by two implementations written
in isolation, and the family of Proposition~\ref{prop:family} are in
\ref{app:stationary}, with what the four have in common:
the stationary vector depends on the letters only through a weighting of
them, and a family of distinct letters realises any rational weighting, so
no hypothesis on the multiset of letters restores the implication.  The
least weighting at which the implication fails is measured there too.

For singletons the implication holds, for every alphabet and every
weighting.

\begin{theorem}
\label{thm:singleton}
Let $\mathcal A$ be synchronizing and strongly connected with $n\ge2$
states, let the letters carry weights $\pi_x>0$, and let $e$ be the
stationary vector of $\sum_x\pi_x\pi(x)$.  If
$\mathrm{minext}(\{q\})>n-1$ then $e_q=\min_pe_p$ and $e_q<e(Q)/n$.  In
particular a stuck singleton has $B_e(\{q\})\le-1$.
\end{theorem}

The proof is in \ref{app:stationary}.

Give a binary automaton the letter weights $(t,1-t)$ and, for each
stuck subset $S$,
let $t^{\ast}(S)$ be the least $t\ge1/2$ at which
$n\,e_t(S)\ge|S|\,e_t(Q)$, either letter being allowed to carry the
weight $t$.  The census of $t^{\ast}$ in \ref{app:stationary} finds that
its minimum is $2/3$ at $n=5$, $6$, $7$ exhaustively, and $2-\sqrt2$ at
$n=8$ and $0.594$ at $n=9$ on the permutation-first strata, where it need
not fall with $n$; at $n=4$ no stuck subset has
$n\,e_t(S)\ge|S|\,e_t(Q)$ at any $t<1$.  At $n=9$ the thresholds were
evaluated on a grid of step $1/400$ and refined exactly for the pairs at
or below the $2/3$ bin, so a pair with $t^{\ast}$ between two grid points
below $0.594$ would not have been seen.  Off those strata it is $1/2$ at $n=9$: a two-letter
automaton with equal letter weights can have a stuck subset with
$B_e(S)\ge0$ (Proposition~\ref{prop:half}), though not a stuck singleton
(Theorem~\ref{thm:singleton}).  In the second witness below
$e(S)/e(Q)=|S|/n$ at every weighting of the letters, by
Lemma~\ref{lem:twovalued}.

\begin{lemma}
\label{lem:twovalued}
Let $Q=Y\sqcup Z$, let $b$ be a permutation of $Q$, and suppose that $a$
maps $Z$ bijectively onto $Z$ and $Y$ bijectively onto $b(Y)$.  Then for
every weighting $(\pi_a,\pi_b)$ of the letters the vector equal to $\pi_b$
on $Y$ and to $\pi_a+\pi_b$ on $Z$ is a stationary vector of
$P=(\pi_a\pi(a)+\pi_b\pi(b))/(\pi_a+\pi_b)$.
\end{lemma}

The proof is in \ref{app:stationary}.

So on such an automaton, when it is strongly connected,
$e(S)/e(Q)=|S|/n$ at every weighting exactly when $|S\cap Z|\,n=|S|\,|Z|$;
every such $S$ has $B_e(S)=0$, so a witness needs only synchronization,
strong connectivity and $S$ stuck.  For prime $n$, and $Y$ and $Z$ both nonempty, no proper $S$
satisfies the condition.

\begin{proposition}
\label{prop:half}
(i) On $Q=\{0,\dots,8\}$ let $a=[8,3,4,0,2,7,6,5,0]$ and
$b=[0,0,6,4,5,5,2,1,5]$.  The automaton is synchronizing and strongly
connected, $S=\{0,1,3,8\}$ has $\mathrm{minext}(S)=9$, and at equal
weights $e=(6,2,1,1,1,8,1,4,3)$, so $B_e(S)=9\cdot12-4\cdot27=0$.
(ii) On $Q=\{0,\dots,9\}$ let $a=[8,5,2,3,4,1,7,7,0,4]$ and
$b=[0,1,6,4,2,9,3,5,7,8]$.  The automaton is synchronizing and strongly
connected, $S=\{4,6,7,9\}$ has $\mathrm{minext}(S)=10$, and
$e(S)=\tfrac25\,e(Q)=\tfrac{|S|}{n}\,e(Q)$ for every weighting of the
letters; in particular $B_e(S)=0$ at equal weights.
\end{proposition}

The verification is in \ref{app:stationary}.  Exhaustively there is no such subset at
$n\le7$; whether one exists at $n=8$, and whether $B_e(S)>0$ can occur at
equal weights, we do not know.

The certified regions of $B$ and $B_e$ are incomparable; moreover
$\{B_e\ge0\}$ is not contained in $C^{\sharp}$, so no positive reweighting
of $\sigma_1,\dots,\sigma_{n-1}$ contains it; the counts, a four-state
witness and the values of $B_e$ on the bounded-deviation family are in
\ref{app:stationary}.

Three further differences separate $B$ from $B_e$.  A positive $e$
exists only under strong connectivity, which Theorem~\ref{thm:main} does
not assume; the family is graded, certifying at length $t$ where $\sigma_t>0$,
while $B_e$ concerns only the length $n-1$; and $\beta^{\ast}$ costs
$O(kn^{2})$ integer additions where $e$ is an exact eigenvector computation.

\subsection{Steinberg's averaging lemma}
\label{sec:averaging}

Steinberg's averaging lemma \cite[Lemma~2]{st} requires an exact invariance
equality; its per-subset refinement is a constant in $\{1,2\}$, a condition
on whether two words in the support separate preimages, not an inequality
on a functional of $S$.  His Theorem 4 extends the Eulerian bound to
pseudo-Eulerian automata.  At the specialisation of
Section~\ref{sec:antecedents}, $P_1$ uniform on $\Sigma^{t}$ and $R=Q$,
the sample space is the words of length $t$ and Steinberg's random variable
is the preimage cardinality $|Su^{-1}|$, which is $|S|$ plus the $x_u$ of
Section~\ref{sec:qcert}.  He computes its expectation, which is $|S|$, and
takes no higher moment; the second moment of Section~\ref{sec:qcert} is
the next moment of the same random variable over the same sample space.

\subsection{What this paper adds}
\label{sec:novelty}

Theorem~\ref{thm:main} is a polynomial-time computable per-subset
certificate for a cardinality-extending word of length at most $n-1$ in
an automaton with an arbitrary in-degree profile.  The field's reference
index of results \cite{vo} lists the Eulerian, pseudo-Eulerian and
quasi-Eulerian classes; no class in it is defined by a bounded in-degree
deviation, and its extension-method entry is a uniform hypothesis, as is
Section 3.4 of Volkov's survey of the field's methods \cite{vo22}, where a
per-subset certificate would appear.

On Kari's identity, each single-length instance $\sigma_t$ of which is the
identity on a power automaton (Section~\ref{sec:antecedents}), the paper
adds four things.  The whole family is read as one certificate, with the
closed form~\eqref{eq:closed}, the region it leaves open cut out by the
$n-1$ conditions $\sigma_t\le0$, and the closure statement of
Proposition~\ref{prop:region} that no reweighting improves it.  The
hypothesis is localised on the Krylov complement
(Theorem~\ref{thm:krylov}).  The strict half is graded sharply by $t$
while the boundary half is not (Section~\ref{sec:grading}).  And
Corollary~\ref{cor:antisym} states the antisymmetry, a statement about
the subset lattice.  These four make up the
per-subset reading of Section~\ref{sec:intro}; the exactness of the
degeneration onto Kari's hypothesis is Theorems~\ref{thm:fibre}
and~\ref{thm:weights}.

\section{The limits of the certificate}
\label{sec:notdo}

\subsection{No reset bound follows}
\label{sec:no-reset}

Within length $n-1$ the method is limited by the arithmetic below,
whatever the certified region.  A further certificate of the shape of
Section~\ref{sec:qcert} may certify more subsets, as the membership test of
\ref{app:qcert} does, but it cannot certify a stuck one, so it cannot lower
the length.

Suppose a non-singleton one-letter fibre $F=\{q\}x^{-1}$ admits a chain
$F=X_0,\dots,X_m=Q$ with $|X_{i+1}|>|X_i|$, $X_{i+1}=X_iu_i^{-1}$ and
$|u_i|\le\ell$.  The sizes strictly increase from $|F|\ge2$ to $n$, so
$m\le n-2$.  Since $X_{i+1}=X_iu_i^{-1}$ says $X_{i+1}\cdot u_i\subseteq X_i$,
concatenating in reverse index order gives $u=u_{m-1}u_{m-2}\cdots u_0$, of
length at most $(n-2)\ell$, with $Q\cdot u\subseteq F$, hence
$Q\cdot ux\subseteq\{q\}$, so
\[\mathrm{rt}(\mathcal A)\ \le\ 1+(n-2)\,\ell.\]
At the length $\ell=n-1$ that this paper certifies, the conclusion would
be $\mathrm{rt}\le n^{2}-3n+3$, which is $n-2$ below $(n-1)^{2}$; since
$C_n$ has the non-singleton one-letter fibre $\{1\}b^{-1}=\{0,1\}$ and
$\mathrm{rt}(C_n)=(n-1)^{2}$ \cite{ce}, that conclusion is false on
\v{C}ern\'y's automata at every $n\ge3$; at $n=3$ it reads
$\mathrm{rt}(C_3)\le3$ against the true value $4$.  Hence on
\v{C}ern\'y's automata no growth chain has every step of length at most
$n-1$, so no certificate within this length can supply such a
chain.  In the vocabulary of Section~\ref{sec:previous}, Theorem~\ref{thm:main}
supplies on $\{B\ge0\}$ the hypothesis of the $\alpha$-extensibility
proposition at $\alpha=(n-1)/n$, restricted to that region; what remains
open is $\{B<0\}$; and obtaining the hypothesis on all of $\{B<0\}$
within this length is impossible, since the reset bound it would give is
false on $C_n$.

The obstruction is this arithmetic alone: on the exhaustive $n=5$
population, 23,965 of the 32,588 automata admit a growth chain every subset of which has $B\ge0$, with every step of
length at most $n-1$.  The display above is conditional, its arithmetic is that of the
regressive-set bound of Section~\ref{sec:previous}, and its hypothesis is
unavailable at $\ell=n-1$; so bounds on $\mathrm{minext}$ do not iterate
to a bound on $\mathrm{rt}$.

\subsection{\texorpdfstring{$B$}{B} is not a potential function}
\label{sec:lyapunov}

The functional $B$ certifies one extension; it is not monotone along
cardinality-preserving moves.  Call a move $S\mapsto Sx^{-1}$ flat if it preserves cardinality.
Along flat moves $B$ is not monotone in either direction, as the automaton
of Example~\ref{ex:cerny} shows for both directions at once: $\{0\}a^{-1}=\{2\}$
is flat and raises $B$ from $-4$ to $1$, while $\{1\}a^{-1}=\{0\}$ is flat
and lowers it from $3$ to $-4$.  So $B$ is not a potential function for
flat moves, in either direction.

\subsection{A strengthening by flat walks fails}
\label{sec:strengthening}

A flat walk is a sequence of flat moves.  The statement that $B(S)\ge0$
implies a flat walk of length at most $n-2$ reaching a subset with an
extending letter is false at $n=5$.  On the automaton with
$Q=\{0,1,2,3,4\}$, $a:0\mapsto1,\ 1\mapsto2,\ 2\mapsto0,\ 3\mapsto0,\
4\mapsto3$ and $b:0\mapsto0,\ 1\mapsto1,\ 2\mapsto4,\ 3\mapsto0,\
4\mapsto1$, which is synchronizing and strongly connected with
$\mathrm{indeg}=(4,3,1,1,1)$ and $\beta^{\ast}=(42,55,-12,-47,-38)$, the
subset $S=\{1,2,4\}$ has $B(S)=5\ge0$, while $Sa^{-1}=\{0,1\}$ shrinks and
$Sb^{-1}=S$; the flat orbit of $S$ is therefore $\{S\}$, so no flat walk of
any length reaches a subset with an extending letter.  Theorem~\ref{thm:main}
is not contradicted: $\mathrm{minext}(S)=2$, by $u=ba$, whose path passes
through the smaller subset $Sa^{-1}$.  Over the exhaustive
binary populations the statement holds at $n=3$ and $n=4$, where the largest
value of $B$ on a subset from which no such walk exists is $-1$, and fails
on 49 subsets at $n=5$, where that largest value is $+5$.  The certificate
guarantees a short extending word; flat reachability is not guaranteed.

\section{Concluding remarks}
\label{sec:conclusion}

Kari's counting identity, taken at every length $t\le n-1$ on an
arbitrary complete automaton and kept as a defect,
yields a family of linear certificates for extension.  Any strictly
positive weighting of the lengths gives a functional of $[S]$, computable
in polynomial time, that certifies extension within $n-1$ wherever it is
non-negative, on every synchronizing complete automaton
(Proposition~\ref{prop:region}), and that vanishes identically exactly on
the Eulerian automata (Theorem~\ref{thm:weights}).  So the Eulerian fibre
is the zero set of every weighting; the weights $k^{\,n-1-t}$ are
chosen for the closed form~\eqref{eq:closed}.  Every such functional is
linear and vanishes on $Q$, so of a subset and its complement at least one
extends within $n-1$ (Corollary~\ref{cor:antisym}).  The linear family is
closed under reweighting; the second moment of the same preimage sizes is
not in it (Section~\ref{sec:qcert}), and at the cost of one factor $n$ it
establishes extension within $n-1$ on 60 to 95 percent of the subsets with
$B(S)<0$, over the populations measured.  Along any sequence of
non-Eulerian automata on which the words of length $n-1$ merge a
fraction of the state pairs bounded below and the largest fraction of
the landings at one state vanishes, the second moment establishes
extension within $n-1$ on all but a vanishing share of them
(Theorem~\ref{thm:share}).  The certificate $B$ certifies half of the
subsets outside $\{B=0\}$, a set of share at most $1/\sqrt N$ for $N$ states with
$\beta^{\ast}_q\ne0$ (Proposition~\ref{prop:half-share}).

Two limits hold on every automaton.  No certificate that reads the multiset of
preimage sizes can reach a stuck subset, so none of these tests lowers
the length $n-1$; and deciding $\mathrm{minext}(S)\le\ell$ exactly is
NP-complete within that length \cite{bfs}.  The family also has a limit:
the Ces\`aro means of the normalised defects converge to
Friedman's centred stationary weight $B_e$
(Proposition~\ref{prop:cesaro}), which is linear, integer-valued,
antisymmetric and zero exactly on the Eulerian automata, like $B$,
certifies singletons (Theorem~\ref{thm:singleton}) and yet no larger
subset in general (Proposition~\ref{prop:be-fails}).  The
limit of the family lies outside it.

Extension within $n-1$, finally, cannot yield a reset bound in general:
iterated along a growth chain it would give
$\mathrm{rt}\le n^{2}-3n+3$, which \v{C}ern\'y's automata violate, so
some of their subsets do not extend within $n-1$
(Section~\ref{sec:no-reset}), while the family of Kisielewicz and
Szyku\l{}a shows that no constant $\alpha$ makes every strongly connected
synchronizing automaton $\alpha$-extensible (\ref{sec:bounded-deviation}).
Within this length the extension method is therefore limited by this
arithmetic, and what remains open is
$\{B<0\}$.  Theorem~\ref{thm:share} bounds the share of it that the
second moment does not certify by the pairs merged and the largest
in-degree at length $n-1$; that share vanishes along any sequence of
automata on which the words of length $n-1$ merge a fraction of the
state pairs bounded below and the largest fraction of the landings at
one state vanishes.  When no state receives more than a bounded
multiple of the mean number of landings, that share vanishes
exponentially fast in $n$.  Conversely, for $n\ge100$, where the words
of length $n-1$ merge fewer than $n^{3/2}\sqrt{V_{n-1}}/425$ pairs on
average,
$V_{n-1}$ the variance of the relative in-degrees at length $n-1$,
more than $2.7$ percent of the subsets fail the test
(Proposition~\ref{prop:barrier}).  The merged fraction is the
probability that a uniformly random word of length $n-1$ merges a
uniformly random pair of distinct states.  Whether it stays above
$n^{-\varepsilon}$ while the largest fraction of the landings at one
state stays below $n^{-\delta}$, for some $\delta>0$ and
$\varepsilon<\delta/2$, on all but a vanishing share of the
synchronizing strongly connected automata on $n$ states, and where the
share of $C^{\sharp}$ tends, is open (Section~\ref{sec:asymptotic}).
Another question we left open
is whether, on the synchronizing Eulerian binary automata,
the constant $n-1$ is attained only at subset sizes coprime to $n$, the
attaining direction being Theorem~\ref{thm:gcd}; with a third letter the
same question has a negative answer at $n=4$
(Section~\ref{sec:zero-fibre}).

\section*{Data availability}

The populations reported are exhaustive enumerations specified completely
by \ref{sec:conventions} and Table~A.2, including the ternary $n=5$
sub-population, whose representatives Table~A.2 fixes; all reported
quantities are defined in the text.  Two reported items are of another
kind and are identified as such where they appear: the family of
\ref{sec:bounded-deviation}, taken from \cite{ks}, and the random sample
of weightings in Table~A.4.  No dataset is deposited, because every
population is determined by its enumeration rule and contains no datum
beyond it.  The programs that produce every count, with a map from each
computed statement to the program and the command that check it, are
released with this paper as the repository
\url{https://github.com/mihelem/extending-words-certificates}.
\section*{Declaration of competing interest}

The author declares no competing interests.

\section*{Funding}

This research did not receive any specific grant from funding agencies in
the public, commercial, or not-for-profit sectors.


\section*{Declaration of generative AI and AI-assisted technologies in the
manuscript preparation process}

During the preparation of this work the author used LLMs with web search in
order to digitise and revise text, to write and cross-check the enumeration
programs, and to locate related literature.  
After using this tool, the author reviewed
and edited the content as needed and takes full responsibility for the
content of the published article.

\appendix
\section{Computational verification}
\label{sec:verification}

Every count in this paper is produced by a program in the released
repository (Data availability), whose README maps each computed
statement to the program and the command that check it, and every
numerical statement is confined to the populations it was computed on.
The enumeration conventions are in \ref{sec:conventions} and the share of
automata on which the hypothesis of Theorem~\ref{thm:main} is nowhere
used in \ref{app:void-share}.  \ref{app:tables} lists the populations
(Table~A.2), which counts rest on one program and which on two
(Table~A.3), the sweeps of the proved statements (Table~A.4) and the
attainment checks (Table~A.5).
\subsection{Conventions for the enumerations}
\label{sec:conventions}

Unless stated otherwise, binary populations are enumerated in the following
quotient convention: the first letter $a$ ranges over one representative of
each conjugacy class of endofunctions of an $n$-element set (19, 47, 130
and 343 classes at $n=4,5,6,7$; sequence A001372 of \cite{oeis}), the
second letter $b$ over all $n^{n}$ maps, followed by the stated filters.  At
alphabet size $k>2$ the remaining $k-1$ letters range over all multisets of
size $k-1$ drawn from the $n^{n}$ maps; at $n=4$ and $k=3$ this gives
395,299 synchronizing strongly connected automata, against 789,358 if the
two letters after the first are ordered.  Every quantity reported in this
paper is invariant under permuting the letters after the first, so every
maximum and every emptiness statement is unchanged by that quotient; shares
are not, because an automaton is weighted by the number of distinct
orderings it admits, which at $k=3$ halves the weight of the automata whose
two last letters differ.  The shares reported at $k>2$ are shares of the
multiset enumeration.  The one population not enumerated in this way is the
ternary $n=5$ sub-population, whose rule is stated in Table~A.2.

Every isomorphism class of binary automaton occurs, since conjugating
$(a,b)$ by a permutation $\pi$ carrying $a$ to its class representative
$a_{0}$ leaves the pair $(a_{0},\pi b\pi^{-1})$ in the enumeration, and
$\mathrm{minext}$, $\mathrm{rt}$, $d$, synchronization and strong
connectivity are invariants of simultaneous conjugation.  In this
convention an isomorphism class contributes $|C(a)|/|\mathrm{Aut}(\mathcal
A)|$ representatives, where $C(a)$ is the centraliser of the $a$-letter in
the symmetric group and $\mathrm{Aut}(\mathcal A)$ is the group of
permutations of $Q$ commuting with every letter, against
$n!/|\mathrm{Aut}(\mathcal A)|$ labelled automata, so a ratio of two such
counts is not the corresponding labelled ratio; reweighting each
representative by $n!/|C(a)|$ recovers the labelled count exactly.  A share
of 0 or 100 percent is convention-invariant, so every emptiness and
universality statement in this paper is unaffected.  The distortion is
largest when the property refers to the letter $a$: the share of automata with
no permutation letter moves from 76.18 percent in this convention to 91.37
percent labelled on the synchronizing strongly connected $\{d\le2\}$
population at $n=6$ (260,461 automata), and on the unfiltered enumeration
at that size from 90.13 to 96.94 percent.  The emptiness fractions of
Section~\ref{sec:hypothesis}, which do not refer to $a$, move by 0.12
points at $n=5$ over all deviations and 1.04 points at $n=6$ on $\{d\le2\}$
(93.28 percent in this convention against 93.16 labelled, and 55.31 against
54.27), and by 1.84 and 2.48 points at $n=3$ and $n=4$ over all deviations
(70.59 against 68.75, and 79.02 against 76.54); the shift is under 2.5
points at every cell computed, though not monotone in $n$.  Percentages
in this paper are shares of the stated enumeration; where the two
conventions differ materially both are given.  Where a count is instead
taken over all transition tables, with no quotient, the text says so.

\subsection{The share of automata on which the hypothesis is void}
\label{app:void-share}

How often the set $\mathcal K^{\perp}\cap\{0,1\}^{Q}$ of
Theorem~\ref{thm:krylov} contains no proper nonempty subset is an empirical
question.  The theorem assumes
neither synchronization nor strong connectivity, but the populations below
are exhaustive synchronizing strongly connected binary populations and
strata of them: the populations of \ref{sec:verification} at
$n=3$, $4$, $5$ and $6$ and the $\{d\le2\}$ strata at $n=5$, $6$ and $7$;
and the share is sensitive to that choice.  Over the
non-Eulerian automata of each of those populations (an Eulerian automaton has $w=0$,
$\mathcal K=0$, the set is the whole cube, and the question is empty):


\begin{center}
\renewcommand{\baselinestretch}{1}\footnotesize
\setlength{\tabcolsep}{2pt}
\begin{tabularx}{\linewidth}{@{}>{\raggedright\arraybackslash}p{2.4cm}*{7}{>{\centering\arraybackslash}X}@{}}
\toprule
 & $n=3$ & $n=4$ & $n=5$, all $d$ & $n=5$, $d\le2$ & $n=6$, all $d$ & $n=6$, $d\le2$ & $n=7$, $d\le2$ \\
\midrule
non-Eulerian automata & 51 & 1,168 & 31,634 & 12,509 & 1,110,301 & 248,233 & 4,873,053 \\
\addlinespace
with no proper nonempty subset having all $\sigma_t=0$ & 36 & 923 & 29,508 & 11,104 & 719,484 & 137,301 & 4,359,891 \\
\addlinespace
fraction & 70.6\% & 79.0\% & 93.3\% & 88.8\% & 64.8\% & 55.3\% & 89.5\% \\
\bottomrule
\end{tabularx}
\end{center}

\noindent{\small\textbf{Table A.1.}\ Automata on which the certificate's
hypothesis is nowhere used.  Over the non-Eulerian automata of each of the
binary populations and strata named above, the number and share having no
proper nonempty subset $S$ with $\sigma_t(S)=0$ for every $t\le n-1$, that
is with $\mathcal K^{\perp}\cap\{0,1\}^{Q}$ containing only $\emptyset$
and $Q$.  For the quotient convention see \ref{sec:conventions}.\par}

The fraction is not monotone in $n$.  Over all deviations it rises from 70.6
percent at $n=3$ to 93.3 percent at $n=5$ and falls to 64.8 percent at
$n=6$; on the stratum $\{d\le2\}$, the only one measured at $n=7$, it falls
from 88.8 percent at $n=5$ to 55.3 percent at $n=6$ and rises again to 89.5
percent at $n=7$.  The range over the cells measured is 55.3 percent
($n=6$, $\{d\le2\}$) to 93.3 percent ($n=5$, all deviations).  
These are shares in
the quotient convention of \ref{sec:conventions}, which also gives the
labelled shares; the two differ by at most 2.5 points at every cell computed.

\subsection{Populations, programs and sweeps}
\label{app:tables}

Table~A.2 lists every population enumerated for this paper.  A population
is named by its size $n$, alphabet size $k$ and filter; unless the table
says otherwise it consists of synchronizing strongly connected complete
automata in the convention of \ref{sec:conventions}, with the first
letter ranging over one representative of each conjugacy class of maps
and the other letters over all maps, and $\{d\le2\}$ or $\{d\le4\}$
denotes the stratum of the stated deviation.  The two unfiltered
enumerations and the three all-transition-table counts are identified as
such.

\begingroup
\renewcommand{\baselinestretch}{1}\footnotesize
\setlength{\tabcolsep}{3pt}
\begin{xltabular}{\linewidth}{@{}>{\raggedright\arraybackslash}p{1.15cm}>{\raggedright\arraybackslash\hsize=1.3\hsize}X>{\raggedright\arraybackslash\hsize=0.85\hsize}X>{\raggedright\arraybackslash\hsize=0.85\hsize}X@{}}
\toprule
$n$, $k$ & filter and convention & automata & used in \\
\midrule
\endhead
3, 2 & synchronizing strongly connected & 59 (8 Eulerian) & Sections 6.2, 6.3, 6.5, 6.6; Tables A.1, D.1 \\
4, 2 & synchronizing strongly connected & 1,240 (72 Eulerian) & Sections 6.2 to 6.6; Appendix C; Tables A.1, D.1 \\
4, 2 & $\{d\le2\}$ & 811 & Section 6.6; Table D.1 \\
5, 2 & synchronizing strongly connected & 32,588 (954 Eulerian) & Sections 6.1 to 6.6, 8.1, 8.3; Appendix C; Tables A.1, D.1 \\
5, 2 & $\{d\le2\}$ & 13,463 & Tables A.1, D.1 \\
6, 2 & synchronizing strongly connected & 1,122,529 (12,228 Eulerian) & Sections 6.3 to 6.6; Appendix C; Tables A.1, D.1 \\
6, 2 & $\{d\le2\}$ & 260,461 & \ref{sec:conventions}; Tables A.1, D.1 \\
7, 2 & $\{d\le2\}$ & 5,065,635 & Tables A.1, D.1 \\
7, 2 & $\{d\le4\}$ & 22,079,307 & Appendix E; Table A.4 \\
7, 2 & synchronizing strongly connected & 42,605,958 & Appendix C \\
8, 2 & first letter a permutation (22 cycle types), second any map & 123,014,054 & Section 7.3; Appendix C \\
9, 2 & first letter a permutation (30 cycle types), second any map & 3,761,587,885 & Section 7.3; Appendix C \\
4, 3 & synchronizing strongly connected, the two letters after the first a multiset & 395,299 (12,303 Eulerian); 789,358 ordered & Sections 6.2, 6.3, 6.5, 6.6; Appendices C, D; Table D.1 \\
4, 3 & $\{d\le2\}$, multiset & 124,893; 249,605 ordered & Section 6.6; Table D.1 (ordered) \\
5, 3 & $\{d\le2\}$, multiset & 18,974,472; 37,946,778 ordered & Section 6.6; Table D.1 (ordered) \\
5, 3 & synchronizing strongly connected, multiset & 138,540,502 & Appendix C \\
5, 3 & sub-population: first letter one of the 7 cycle types of a permutation, its cycles on consecutive states; second letter all $5^{5}=3{,}125$ maps; third letter the lexicographically least map of each of the 47 conjugacy classes & 1,028,125 triples, 580,893 synchronizing strongly connected & Section 6.3; Table A.4 \\
4, 3 & Eulerian, synchronizing strongly connected; first letter up to conjugacy, the other two ordered & 24,606 & Section 5.2 \\
4 to 9, 2 & synchronizing Eulerian (the Eulerian fibres) & 72; 954; 12,228; 192,582; 3,212,088; 59,605,126 & Section 5.2; Table 1 \\
5, 2 & unfiltered: 31,895 not synchronizing, 82,392 synchronizing but not strongly connected, 32,588 synchronizing strongly connected & 146,875 & Section 6.1; Table A.4 \\
6, 2 & unfiltered & the enumeration of \ref{sec:conventions} & \ref{sec:conventions} \\
3 and 4, 2; 3, 3 & all transition tables, no quotient & 639 non-Eulerian of $3^{6}$; 63,016 of $4^{8}$; 18,003 of $3^{9}$ & Section 5.4 \\
\bottomrule
\end{xltabular}
\endgroup

\noindent{\small\textbf{Table A.2.}\ The populations.  Counts in the
quotient convention of \ref{sec:conventions} unless the row says
otherwise; ``ordered'' counts the two letters after the first as an
ordered pair.\par}

Table~A.3 records, for every computed quantity, the program of the
released repository that produces it and whether a second program
sharing no code with the first reproduced it.  The historical record
behind some of these counts is longer, but only what the repository
reproduces is claimed here.

\begingroup
\renewcommand{\baselinestretch}{1}\footnotesize
\setlength{\tabcolsep}{3pt}
\begin{xltabular}{\linewidth}{@{}>{\raggedright\arraybackslash\hsize=0.9\hsize}X>{\raggedright\arraybackslash\hsize=1.0\hsize}X>{\raggedright\arraybackslash\hsize=1.1\hsize}X@{}}
\toprule
quantity (where reported) & program & second program \\
\midrule
\endhead
binary censuses at $n\le6$ and at $n=7$ on $\{d\le4\}$ (Sections 5.4, 6.1 to 6.5, 8.1, 8.3, Appendix E, Table A.4) & \nolinkurl{census/census.py}, \nolinkurl{census_d4.py}, \nolinkurl{grading.py} & the C engine \nolinkurl{second_moment/census_engine.c} reproduces the counts of Sections 6.3 to 6.5 on the populations of Table D.1 \\
second moment (Sections 4.1, 6.6; Tables A.1, D.1) & \nolinkurl{second_moment/census_engine.c} & \nolinkurl{second_moment/independent_engine} (Python, no shared code): 96 of 96 integer fields agree on binary $n=3,4,5$ and ternary $n=4$, in both conventions \\
membership test (Appendix D) & \nolinkurl{second_moment/membership.py} at $n\le5$; \nolinkurl{membership.rs}, a port of it, at $n=6$ & none \\
ternary populations (Sections 6.2, 6.3, 6.5; Appendix D; Table A.4) & \nolinkurl{ternary/ternary.py} & the C engine on the ternary $n=4$ population and the two ternary strata, in the ordered convention \\
Eulerian fibres (Section 5.2, Table 1) & \nolinkurl{eulerian_fibres/fibres.c} & \nolinkurl{census/census_d4.py} at $n\le7$; none at $n=8$ and $n=9$ \\
the witnesses of Propositions~\ref{prop:be-fails} and~\ref{prop:half} (Appendix C) & \nolinkurl{stationary/witnesses.py}, \nolinkurl{equal_weight_witnesses.py}, exact rational arithmetic & none \\
the family of Proposition~\ref{prop:family} at $n\le40$ (Appendix C) & \nolinkurl{stationary/family_witnesses.py}, exact rational arithmetic & none; the proposition is proved \\
uniform-weight censuses (Appendix C) & \nolinkurl{stationary/uniform_weights.c} & \nolinkurl{permutation_stratum.c} agrees on the automaton and stuck-subset counts and on the maximum on every binary population up to $n=8$; the $n=9$ stratum and the ternary $n=5$ population rest on one program each \\
incomparability of $\{B\ge0\}$ and $\{B_e\ge0\}$ (Appendix C) & \nolinkurl{stationary/uniform_weights.c} & none \\
census of $t^{\ast}$ at $n\le7$ (Section 7.3, Appendix C) & \nolinkurl{stationary/bias_census.c} & \nolinkurl{bias_census_exact.c}: every histogram bin agrees; the 133 attaining pairs re-verified exactly by \nolinkurl{boundary_classes.py}, with $2/3$ an exact root \\
$n=8$ stratum thresholds (Appendix C) & \nolinkurl{stationary/bias_census_exact.c} & the 448 pairs at or below $2/3$ re-verified exactly by \nolinkurl{exact_n8.py} \\
$n=9$ stratum thresholds (Appendix C) & \nolinkurl{stationary/permutation_stratum.c}, on the grid of step $1/400$, after reproducing the $n=8$ stratum digit for digit & the 1,984 pairs at or below the $2/3$ bin re-verified exactly by \nolinkurl{exact_thresholds.py} \\
$L=135$ (Section 7.3) & \nolinkurl{stationary/berlinkov_L.py} & none \\
identity~\eqref{eq:defect} and $\{\beta^{\ast}=0\}$ (Table A.4) & \nolinkurl{census/identity_check.py}; \nolinkurl{census_d4.py}; \nolinkurl{ternary/ternary.py} & checks of the statements, one program each \\
root location (Table A.5) & \nolinkurl{examples/root_location.py}; the random weightings by \nolinkurl{census/census_d4.py} & none \\
examples, flat walks, the family of Appendix E & \nolinkurl{examples/}, \nolinkurl{ks_family/} & none; each program checks the statements it prints \\
Eulerian ternary automata at $n=4$ (Section 5.2) & \nolinkurl{eulerian_fibres/ternary_converse_n4.py} & none \\
\bottomrule
\end{xltabular}
\endgroup

\noindent{\small\textbf{Table A.3.}\ Programs and agreement.  Paths are
relative to the released repository.\par}

Theorem~\ref{thm:main} and Corollaries~\ref{cor:strict}
and~\ref{cor:antisym} are proved, so their sweeps test the programs; the
sweeps of the non-synchronizing automata test the statements, since
there the boundary case can fail and does.  Table~A.4 lists them.  No
never-extendable subset occurs on any strongly connected population,
where a reset word followed by a path into $S$ gives $Sw^{-1}=Q$ for
every proper nonempty $S$; on the synchronizing but not strongly
connected automata such subsets do occur, every one of them with
$B(S)<0$ as Theorem~\ref{thm:main} forces.

\begin{center}
\renewcommand{\baselinestretch}{1}\footnotesize
\setlength{\tabcolsep}{3pt}
\begin{tabularx}{\linewidth}{@{}>{\raggedright\arraybackslash}X>{\raggedright\arraybackslash}X>{\raggedleft\arraybackslash}p{1.7cm}>{\raggedright\arraybackslash}p{2.5cm}>{\raggedright\arraybackslash}p{2.1cm}@{}}
\toprule
statement & population & automata & stuck subsets & violations \\
\midrule
Theorem~\ref{thm:main} & binary $n=5$ & 32,588 & 11,042 & 0 \\
 & binary $n=6$ & 1,122,529 & 186,497 & 0 \\
 & binary $n=7$, $\{d\le4\}$ & 22,079,307 & 1,853,891 & 0 \\
 & ternary $n=4$ & 395,299 & 43,894 & 0 \\
 & ternary $n=5$ sub-population & 580,893 & & 0 \\
 & in all & 24,210,616 & & 0 \\
Corollary~\ref{cor:strict} & binary $n=5$, unfiltered & 146,875 & & 0 \\
boundary case $B(S)=0$ & binary $n=5$, synchronizing not strongly connected & 82,392 & & 0 \\
boundary case $B(S)=0$ & binary $n=5$, not synchronizing & 31,895 & & 58,060 \\
Corollary~\ref{cor:antisym}: pairs stuck on both sides & binary $n=5$; $n=6$; $n=7$, $\{d\le4\}$ & & 11,042; 186,497; 1,853,891 one-sided & 0 \\
identity~\eqref{eq:defect} against word enumeration & binary $n=5$ (977,640 instances); ternary $n=4$ & & & 0 \\
$\{\beta^{\ast}=0\}$ is the Eulerian fibre, both directions & binary $n=4$ to $7$ (72; 954; 12,228; 192,582 with $\beta^{\ast}=0$); ternary $n=4$; ternary $n=5$ sub-population & & & 0 \\
Proposition~\ref{prop:family} & the family, $n=5$ to $40$ & 36 & 36 & 0 \\
Theorem~\ref{thm:weights}: no root of $q_c$ with $|z|\le k$ & 1,760 random weightings, $k=2,\dots,5$, $n=3,\dots,13$ & & & 0 (least $\min|z|/k$ observed 1.00075) \\
\bottomrule
\end{tabularx}
\end{center}

\noindent{\small\textbf{Table A.4.}\ Sweeps of the proved statements.
The 58,060 violations on the non-synchronizing automata are the proper
nonempty members of the flat locus of Section~\ref{sec:qcert}.\par}

The sharpness of the root bound of Theorem~\ref{thm:weights} needs no
sampling: taking $c=(\varepsilon,\dots,\varepsilon,1)$ makes $q_c$ tend
coefficientwise to $h_{n-1}$, at fixed degree $n-2$ since
$\gamma_{n-2}=c_{n-1}=1$, so its roots converge to those of $h_{n-1}$,
which lie on $|z|=k$; thus $\min|z|/k$ tends to $1$ while
Theorem~\ref{thm:weights} holds it strictly above $1$ at every
$\varepsilon>0$.  So at each fixed $n\ge3$ and $k$ the infimum of
$\min|z|/k$ over the positive cone of weightings is exactly $1$ and is not
attained: no bound $\min|z|/k\ge\rho>1$ with $\rho$ depending only on $n$
and $k$ is available, though weight-dependent ones are,
Enestr\"om--Kakeya giving $\min|z|/k\ge\min_jg_j/g_{j+1}$.  Table~A.5
collects the attainment checks and the remaining measured values.

\begin{center}
\renewcommand{\baselinestretch}{1}\footnotesize
\setlength{\tabcolsep}{3pt}
\begin{tabularx}{\linewidth}{@{}>{\raggedright\arraybackslash}X>{\raggedright\arraybackslash}X>{\raggedright\arraybackslash}X@{}}
\toprule
statement & population & value \\
\midrule
$n-1$ attained inside $\{B\ge0\}$ (Section 6.3) & ternary $n=4$; ternary $n=5$ sub-population & 11,385 and 9,698 subsets, at every subset size \\
maximum of $\mathrm{minext}$ over singletons off the Eulerian fibre & binary $n=5$, by deviation $d=2,4,6,8$ & $n$ at each \\
maximum reset threshold, as a check against \cite{cpr} & binary $n=4$, $5$, $6$; $n=7$ on $\{d\le4\}$ & 9, 16, 25, 36, at $d=2$ only, by the \v{C}ern\'y automata \\
$\#\{B<0\}=\#\{B>0\}$, as $B(Q\setminus S)=-B(S)$ requires & binary $n=5$ & 463,047 \\
$\min|\mathrm{root}(p)|/k$ (independent of $k$) & $n=3,\dots,40$ & falls strictly from $2$ to $1.05427$, above $1$ throughout \\
$\min|z|/k$ at $c=(\varepsilon,\dots,\varepsilon,1)$ & $n=13$, $k=2$, $\varepsilon=10^{-6}$ & $1.00000004$ \\
the coprimality conjecture with three letters (Section 5.2) & Eulerian ternary $n=4$ & the maximum of $\mathrm{minext}$ is $3$ at every subset size; 832 pairs (automaton, subset) of size $2$ attain it \\
\bottomrule
\end{tabularx}
\end{center}

\noindent{\small\textbf{Table A.5.}\ Attainment checks and the remaining
measured values.\par}
\section{Proofs of the fibre results}
\label{app:fibre-proofs}

\begin{proof}[Proof of Theorem~\ref{thm:gcd}]
Write $\ell=n-k$, $x_0=0$, $y_0=\ell$, $t_i=i$ for $i<\ell$ and
$c_j=\ell+j$ for $0\le j<k$: the letter $b$ maps the tail
$t_0\to t_1\to\dots\to t_{\ell-1}\to c_0$ onto the cycle
$c_0\to\dots\to c_{k-1}\to c_0$.  The state $x_0$ has the $a$-preimages
$x_0,y_0$ and no $b$-preimage, $y_0$ has the $b$-preimages
$t_{\ell-1},c_{k-1}$ and no $a$-preimage, and every other state has one
preimage under each letter; so the automaton is Eulerian, and strongly
connected since the $b$-orbit of $x_0$ is $Q$ and $a(y_0)=x_0$.

If $d=\gcd(k,n)>1$ then $d$ divides $\ell$, and the partition of $Q$ by
residues modulo $d$ is a congruence: $b$ maps the class of $i$ to the
class of $i+1$, also at $i=n-1$ since $n\equiv\ell\pmod d$, and $a$
maps $\ell$ to $0$.  Since the letters permute the classes of this
partition and the classes have equal size, every union of classes has
the same number of preimages under every word, so a proper nonempty
union never extends, against Theorem~\ref{thm:main}; hence the automaton
is not synchronizing.  If
$\gcd(k,n)=1$, give the cycle state $c_j$ the phase $j\in\mathbb Z_k$.
The letter $b$ advances every phase by one and brings every tail state
onto the cycle within $\ell$ steps; $a$ moves only $y_0=c_0$, to $x_0$,
which re-enters the cycle at phase $0$ after $\ell$ further steps of $b$,
while every other cycle state advances by $\ell\equiv n\pmod k$.  So for
two cycle states at phase difference $\Delta$ the word $b^{j}ab^{\ell}$,
with $j$ chosen to bring one of them to phase $0$, changes $\Delta$ by
$\pm n$ modulo $k$; as $n$ is invertible modulo $k$, every difference is
reduced to $0$, every pair of states merges, and the automaton is
synchronizing.

Call $T\subseteq Q$ balanced if $|Ta^{-1}|=|Tb^{-1}|=|T|$.  On an Eulerian
binary automaton the two preimage sizes sum to $2|T|$, so $T$ is balanced
if and only if it contains as many states with two $a$-preimages as
states with two $b$-preimages, here if and only if $x_0\in T\Leftrightarrow
y_0\in T$.  On an Eulerian automaton the average of $|Su^{-1}|$ over the
words of length $t$ is $|S|$, so $\mathrm{minext}(S)\ge n-1$ if and only if
$|Su^{-1}|=|S|$ for every $|u|\le n-2$, that is, if and only if every
$Su^{-1}$ with $|u|\le n-3$ is balanced; and then
$\mathrm{minext}(S)=n-1$ by Theorem~\ref{thm:main}.  A balanced $T$
satisfies $Ta^{-1}=T$, and $Tb^{-1}$ is the set of $b$-predecessors of the
states of $T$, both predecessors of $y_0$ included when $y_0\in T$.  Hence,
by induction on the length, if $Sb^{-s}$ is balanced for every $s<t$
then every $Su^{-1}$ with $|u|\le t$ is one of the sets $Sb^{-s}$ with
$s\le|u|$; and $Sb^{-t}$ is balanced if and only if
\begin{equation}
\label{eq:ct}
x_0\cdot b^{t}\in S\quad\Longleftrightarrow\quad y_0\cdot b^{t}\in S,
\end{equation}
so that $\mathrm{minext}(S)=n-1$ if and only if~\eqref{eq:ct} holds for
$0\le t\le n-3$.  Now $x_0\cdot b^{t}=t_t$ for $t<\ell$ and
$x_0\cdot b^{t}=c_{(t-\ell)\bmod k}$ for $t\ge\ell$, while
$y_0\cdot b^{t}=c_{t\bmod k}$.  Put $A=\{j:c_j\in S\}\subseteq\mathbb Z_k$
and $\nu=n\bmod k$.  For $t<\ell$, where $\ell-1\le n-4$ because $k\ge3$,
condition~\eqref{eq:ct} says $t_t\in S\Leftrightarrow t\bmod k\in A$: the
tail part of $S$ is determined by $A$.  For $\ell\le t\le n-3$, with
$j=t-\ell$ running over $0,\dots,k-3$, it says $j\in A\Leftrightarrow
j+\nu\in A$.  Since $\gcd(\nu,k)=1$ the map $j\mapsto j+\nu$ is a single
cycle on $\mathbb Z_k$, and the conditions for $j=0,\dots,k-3$ are all of
its edges except the two leaving $k-2$ and $k-1$; so $A$ is a union of the
two arcs $A_1,A_2$ into which those two edges cut the cycle, and $S$
proper nonempty means $A\in\{A_1,A_2\}$.  The two subsets so obtained are
complementary, since their cycle parts and their tail parts are.

For the sizes write $n=qk+\nu$, so $\ell=(q-1)k+\nu$, and let
$e=\nu^{-1}\bmod k$.  The arc ending at $k-1$ is
$A_2=\{-1-t\nu\bmod k:0\le t<e\}$, of $e$ elements, and for either arc
$|S|=|A|+\#\{i<\ell:i\bmod k\in A\}=q|A|+|A\cap[0,\nu)|$.  The residue
$-1-t\nu$ lies in $[0,\nu)$ exactly when the interval $[t\nu+1,(t+1)\nu]$
contains a multiple of $k$; for $0\le t<e$ these intervals, each of length
$\nu<k$, tile $[1,e\nu]=[1,1+hk]$, where $e\nu=1+hk$, which contains
exactly the $h$ multiples $k,\dots,hk$.  Hence $|A_2\cap[0,\nu)|=h$ and
$|S_{A_2}|=qe+h=(e(qk+\nu)-1)/k=(en-1)/k$, so $|S_{A_2}|\,k\equiv-1$ and
$|S_{A_1}|\,k=(n-|S_{A_2}|)\,k\equiv1\pmod n$.

Finally, given $n\ge4$ and $m$ coprime to $n$, the two residues $\pm m^{-1}$
modulo $n$ are coprime to $n$ and sum to $n$, so one of them is at least
$3$ (at $n=4$ neither can be $2$); for that $k$ the automaton
$\mathcal E(n,k)$ has attaining subsets of the sizes $\pm k^{-1}\equiv\pm m$
modulo $n$, that is, of sizes $m$ and $n-m$.
\end{proof}

\begin{proof}[Proof of Proposition~\ref{prop:gcd-partial}]
Let $p$ be a prime dividing both $|S|$ and $n$, and work over the field
$\mathbb F_p$.  Put $f=[S]^{\mathsf T}$, which has entries $0$ and $1$
and is not a multiple of $\mathbf 1$, so
$W_0=\mathrm{span}\{\mathbf1,f\}$ has dimension $2$; as
$\pi(x)\mathbf1=\mathbf1$, the space $W_t$ of Lemma~\ref{lem:stabilise}
is spanned by $\mathbf1$ and the vectors $\pi(u)f=[Su^{-1}]^{\mathsf T}$
with $|u|\le t$.  On an Eulerian automaton every $\sigma_t$ vanishes,
so by Lemma~\ref{lem:defect} the average of $|Su^{-1}|$ over the words
of length $t$ is $|S|$, and $\mathrm{minext}(S)=n-1$ gives
$|Su^{-1}|=|S|$ for every $|u|\le n-2$.  Hence every generator of
$W_{n-2}$ has coordinate sum $|S|$ or $n$, both zero in $\mathbb F_p$,
and Lemma~\ref{lem:stabilise} with $H$ the hyperplane of vectors with
zero coordinate sum, $h=n-1$ and $m=2$, gives $\pi(u)f\in H$ for every
word $u$: $p$ divides $|Su^{-1}|$ for every word $u$.  For every letter
$x$ and every $T\subseteq Q$,
\[|Tx^{-1}|-|T|=\sum_{q\in T}\bigl(|qx^{-1}|-1\bigr);\]
the negative summands, one for each state of $T$ with no $x$-preimage,
number at most $n-|Q\cdot x|=\kappa_x$, and the positive ones total at
most $\sum_{q}\max(0,|qx^{-1}|-1)$, which is again $\kappa_x$ because
$\sum_q(|qx^{-1}|-1)=0$; so the difference is an integer of absolute
value at most $\kappa_x$.  Applied to $T=Su^{-1}$, the difference
$|S(xu)^{-1}|-|Su^{-1}|$ is divisible by $p$ and at most $\kappa_x$ in
absolute value; if $p>\kappa_x$ for every letter $x$, all these
differences vanish, so $|Su^{-1}|=|S|$ for every word, which a reset
word contradicts.  Hence $p\le\max_x\kappa_x$.  On a binary Eulerian
automaton every state has two preimages in all, so a state with two
$a$-preimages has no $b$-preimage and conversely; as $|qa^{-1}|\le2$,
$\kappa_a=\sum_q\max(0,|qa^{-1}|-1)$ is the number of states with two
$a$-preimages, which is the number with no $b$-preimage, $\kappa_b$.
\end{proof}

\section{The stationary-vector functional: proofs, witnesses and censuses}
\label{app:stationary}

This appendix collects the proofs and the measurements behind
Section~\ref{sec:stationary}.

\begin{proof}[Proof of Proposition~\ref{prop:cesaro}]
$M$ is irreducible with row sums $k$, so by the Perron--Frobenius theorem
its peripheral eigenvalues are $k\omega$ for the $h$-th roots of unity
$\omega$, $h$ the period, each simple, and every other eigenvalue has
modulus strictly below $k$.  Hence
$M^{t}k^{-t}=\sum_{\omega^{h}=1}\omega^{t}\,\Pi_{\omega}+R_t$ with
$\|R_t\|\to0$ geometrically, where $\Pi_{\omega}$ is the spectral projection
at $k\omega$ and
$\Pi_1=\mathbf1e^{\mathsf T}/\langle\mathbf1,e\rangle$, since $\mathbf1$ and
$e^{\mathsf T}$ are the right and left eigenvectors at $k$.  So
\[\frac{\sigma_t(S)}{k^{t}}
=\frac{\mathbf1^{\mathsf T}M^{t}[S]^{\mathsf T}}{k^{t}}-|S|
=\frac{n\,\langle[S],e\rangle-|S|\,\langle\mathbf1,e\rangle}
{\langle\mathbf1,e\rangle}
+\sum_{\omega\ne1}\omega^{t}\,\mathbf1^{\mathsf T}\Pi_{\omega}[S]^{\mathsf T}
+o(1),\]
whose constant term is $B_e(S)/\langle\mathbf1,e\rangle$; the oscillating
terms average to zero over $t$, and are absent when $h=1$.
\end{proof}

\begin{proof}[Proof of Proposition~\ref{prop:be-fails}]
Letters are written as the list of images of $0,1,\dots,n-1$.  Each
witness was checked with exact rational arithmetic by two implementations
written in isolation: the integer eigenvector is recomputed from
$e^{\mathsf T}M=ke^{\mathsf T}$, strong connectivity and synchronization
by breadth-first search, and $\mathrm{minext}(S)$ by breadth-first search
over all preimage sets.

On $Q=\{0,\dots,4\}$ let $a=[0,2,1,1,3]$ and $c=[3,1,4,0,2]$, so that $c$
is the permutation $(0\,3)(2\,4)$.  With the three letters $a,a,c$ the
automaton is synchronizing and strongly connected, $e=(1,4,3,1,1)$, and
$S=\{2,4\}$ has $B_e(S)=5\cdot4-2\cdot10=0$ while
$\mathrm{minext}(S)=5$: the chain
$\{2,4\}\to\{1\}\to\{2,3\}\to\{1,4\}\to\{1,2\}\to\{1,2,3\}$ of preimages
under $a,a,a,c,a$ is a shortest extension.  With $a$ taken three times,
$e=(1,5,4,1,1)$ and $B_e(S)=1$; with $a$ taken once, $e=(1,3,2,1,1)$ and
$B_e(S)=-1$.

On $Q=\{0,\dots,6\}$ let $a=[0,0,2,4,3,3,5]$, $b=[1,4,5,3,6,2,0]$ and
$a'=[0,5,2,4,3,3,0]$, so that $a'=a\circ(1\,6)$.  The three letters are
pairwise distinct, the automaton is synchronizing and strongly connected,
$e=(3,1,1,4,3,1,1)$, and $S=\{0,1\}$ has $B_e(S)=7\cdot4-2\cdot14=0$ with
$\mathrm{minext}(S)=7$.  Finally let $x_1=[5,3,2,4,4,0,6]$,
$x_2=[1,0,0,2,3,6,5]$, $x_3=[0,3,2,4,4,5,6]$, $x_4=[0,3,2,4,4,6,5]$ and
$x_5=[5,3,2,4,4,6,0]$, five pairwise distinct letters on seven states.
The automaton is synchronizing and strongly connected,
$e=(5,1,4,4,16,5,5)$, and $S=Q\setminus\{6\}$ has
$B_e(S)=7\cdot35-6\cdot40=5>0$ and $\mathrm{minext}(S)=7$.
\end{proof}

All four witnesses are instances of one construction.  Give the letters weights
$\pi_x>0$ summing to $1$ and let $e_\pi$ be the stationary vector of
$P=\sum_x\pi_x\pi(x)$, normalised to $e_\pi(Q)=1$; the uniform weighting
recovers $e$ up to scale, and a letter occurring $r$ times among $k$
carries weight $r/k$.  The first witness is therefore the two-letter
automaton $(a,c)$ with weights $(2/3,1/3)$.  At weights $(t,1-t)$ its
stationary vector is proportional to $(1-t,\,2-t,\,1,\,1-t,\,1-t)$ and
$5\,e_t(S)-2$ is a positive multiple of $3t-2$: the subset $S$ is stuck at
every $t$ and carries at least the average weight per state exactly when
$t\ge2/3$.  Repeating a letter is one way to reweight the letters; the
second witness reweights without repetition.
In the second witness the states $1$, $2$, $5$ and $6$ have equal
stationary weight at every $t$, so $e_t\pi(a\circ\sigma)=e_t\pi(a)$ for
every permutation $\sigma$ of these four states: the letter $a'$ leaves
the stationary vector unchanged, though not the preimage tree of $S$, so
the three-letter automaton has the stationary vector of $(a,b)$ at
weights $(2/3,1/3)$.  The third witness does the same with four
letters that act on $e$ like one letter against a fifth, weights
$(4/5,1/5)$, where $B_e(S)$ is strictly positive.  So no hypothesis
on the multiset of letters, such as pairwise distinctness or a bound on
multiplicities, restores the implication: the stationary vector depends
on the letters only through their weights, and a family of
distinct letters realises any rational weighting whenever the extra
letters do not shorten the extension of $S$.  The least weighting at
which the implication fails is measured below.

One construction gives a witness on every number of states from five
on.

\begin{proposition}
\label{prop:family}
For $n\ge5$ write $n=2L+1$ or $n=2L+2$ with $L\ge2$, put $r=n-L-2$,
and let $\mathcal A_n$ be the automaton on the
states $\gamma_1,\dots,\gamma_L,h_1,h_2,z_1,\dots,z_r$ whose letter $a$
sends $\gamma_i$ to $\gamma_{i+1}$ with indices modulo $L$, $h_2$ to
$h_1$, $h_1$ to $\gamma_1$ and fixes every $z_i$, and whose letter $c$
is the permutation $(\gamma_L\,h_2)(\gamma_{L-1}\,h_1\,z_1\cdots z_r)$.
The automaton $\mathcal A_n$ is synchronizing and strongly connected,
$S=\{\gamma_L,h_2\}$ has $\mathrm{minext}(S)=2L+2\ge n$, and with the
weight $t$ on $a$ and $1-t$ on $c$ the stationary vector normalised by
$e_t(\gamma_L)=1$ satisfies
\[
n\,e_t(S)-|S|\,e_t(Q)=(n-2L)-(1-t)(3n-8).
\]
The inequality $n\,e_t(S)\ge|S|\,e_t(Q)$ holds if and only if
$t\ge1-(n-2L)/(3n-8)$, with equality at the threshold, where $B_e(S)=0$: the letter $a$ taken
$3n-9$ times against one $c$ realises the threshold for odd $n$, and
taken $(3n-10)/2$ times for even $n$.
\end{proposition}

\begin{proof}
Write $\Gamma=\{\gamma_1,\dots,\gamma_L\}$ and $Z=\{z_1,\dots,z_r\}$.
Every state reaches $\gamma_1$, the states of $\Gamma\cup\{h_1,h_2\}$
by $a$ and each $z_i$ by $c$ through $\gamma_{L-1}$; and $\gamma_1$
reaches every $\gamma_i$ by $a$, $h_2$ by $c$ from $\gamma_L$, $h_1$ by
$a$ from $h_2$ and each $z_i$ by $c$ from $h_1$.  The automaton is
therefore strongly connected.  It is synchronizing because every pair of states
is merged by some word.  The word $a^{2}$ sends a pair into
$\Gamma\cup Z$.  The word $acaa$ sends $\{\gamma_{L-1},z_r\}$ to
$\{\gamma_1\}$ and $c\,a^{L-1}$ sends $\{\gamma_{L-1},z_j\}$ to
$\{\gamma_{L-1},z_{j+1}\}$ for $j<r$, so every pair
$\{\gamma_{L-1},z_j\}$ is merged.  A pair $\{\gamma_k,z_j\}$ becomes one of these under a power
of $a$, and a pair $\{z_i,z_j\}$ under a power of $c$; a pair
$\{\gamma_k,\gamma_{L-1}\}$ with $k\le L-2$ becomes
$\{\gamma_{L-1},z_1\}$ under $c\,c\,a^{L-1-k}$, the pair
$\{\gamma_{L-1},\gamma_L\}$ becomes it under $c\,c\,a^{L-1}$, and a
power of $a$ sends every pair in $\Gamma$ to one of these two.

For a set $T$ the preimage $Tc^{-1}$ has the size of $T$ and
$|Ta^{-1}|=|T|+[\gamma_1\in T]-[h_2\in T]$, since $\gamma_1$ is the
only state with two $a$-preimages and $h_2$ the only one with none.
Call a split an application of $a^{-1}$ to a set containing $\gamma_1$
and not $h_2$; only a split raises the size, and an extension of $S$
is a split of a set of size $2$.  In the digraph on $Q$ with an edge
$p\to p'$ whenever $p'\cdot x=p$ for some letter $x$, let $d(p)$ be
the distance from $p$ to $\gamma_1$; one letter lowers
$\min_{p\in T}d(p)$ by at most one, as every state of $Tx^{-1}$ is a
preimage of a state of $T$.  The preimages of $\gamma_i$ for
$2\le i\le L-2$ are $\gamma_{i-1}$ and $\gamma_i$, those of
$\gamma_{L-1}$ are $\gamma_{L-2}$ and $z_r$ when $L\ge3$, those of $\gamma_L$ are
$\gamma_{L-1}$ and $h_2$, those of $h_1$ are $h_2$ and $\gamma_{L-1}$,
that of $h_2$ is $\gamma_L$, those of $z_1$ are $z_1$ and $h_1$ and
those of $z_i$ are $z_i$ and $z_{i-1}$.  The distances are therefore
$d(\gamma_i)=i-1$, $d(h_1)=L-1$, $d(h_2)=L$ and $d(z_i)=L-1+i$.  As
$Sc^{-1}=S$ and
$Sa^{-1}=\{\gamma_{L-1}\}$, every extending word passes through
$\{\gamma_{L-1}\}$, and its first split, of $\{\gamma_1\}$ into
$\{\gamma_L,h_1\}$, comes after at least $1+(L-2)+1=L$ letters.  Let
$s\ge L$ count the letters up to the last split before the extension,
which turns $\{\gamma_1\}$ into $\{\gamma_L,h_1\}$.  Both letters send
$\{\gamma_L,h_1\}$ to $\{\gamma_{L-1},h_2\}$; let $s'\ge s+1$ count the
letters up to the last visit to $\{\gamma_{L-1},h_2\}$ before the
extension.  There $a$ gives $\{\gamma_{L-2}\}$ if $L\ge3$, a singleton
that must split again before any extension, against the choice of $s$,
and $\{\gamma_L,h_1\}$ if $L=2$, from which both letters lead back to
$\{\gamma_{L-1},h_2\}$, against the choice of $s'$.  The letter $c$
therefore follows, giving $\{z_r,\gamma_L\}$, at distance $L-1$ from $\gamma_1$, and the
extension comes after at least $s'+1+(L-1)+1\ge2L+2$ letters.  The
preimages under $a$ taken $L+1$ times, then $c$, then $a$ taken $L$
times, run through $\{\gamma_{L-1}\},\dots,\{\gamma_1\}$,
$\{\gamma_L,h_1\}$, $\{\gamma_{L-1},h_2\}$, $\{z_r,\gamma_L\}$,
$\{z_r,\gamma_{L-1}\},\dots,\{z_r,\gamma_1\}$ and end at
$\{z_r,\gamma_L,h_1\}$; so $\mathrm{minext}(S)=2L+2$, which is $n+1$
for odd $n$ and $n$ for even $n$.

For the stationary vector write $u=1-t$, so that
$e_q=t\,e(qa^{-1})+u\,e(qc^{-1})$ at every state $q$.  At $q=z_i$ this
gives $e(z_i)=e(h_1)$; at $q=h_2$, $e(h_2)=u\,e(\gamma_L)$; at
$q=\gamma_L$, $(1-u^{2})e(\gamma_L)=t\,e(\gamma_{L-1})$, that is
$e(\gamma_{L-1})=(1+u)e(\gamma_L)$; at $q=h_1$,
$e(h_1)=t\,e(h_2)+u\,e(\gamma_{L-1})=2u\,e(\gamma_L)$; at $q=\gamma_i$
for $2\le i\le L-2$, $e(\gamma_i)=e(\gamma_{i-1})$; and at $q=\gamma_1$
for $L\ge3$, $e(\gamma_1)=e(\gamma_L)+e(h_1)$.  With
$e(\gamma_L)=1$ the vector is $1+2u$ on $\gamma_1,\dots,\gamma_{L-2}$,
$1+u$ at $\gamma_{L-1}$, $u$ at $h_2$ and $2u$ at $h_1$ and at every
$z_i$, whence $e(Q)=L+2(n-2)u$, $e(S)=1+u$ and
$n\,e(S)-2\,e(Q)=(n-2L)-(3n-8)u$.  A letter occurring $m$ times among
$k$ carries weight $m/k$, which gives the two multiplicities.
\end{proof}

Table~A.4 records the exact check of the proposition for $n=5$ to
$40$.

\begin{proof}[Proof of Theorem~\ref{thm:singleton}]
Call a state fat if some letter has at least two preimages at it and thin
otherwise, and let $F$ be the set of fat states; $F\ne\emptyset$, since
otherwise every letter is a permutation and no proper nonempty subset ever
extends.  In the digraph on $Q$ with an edge $p\to p'$ whenever
$p'\cdot x=p$ for some letter $x$, let $d(p)$ be the distance from $p$ to
$F$, finite by strong connectivity.  Then
$\mathrm{minext}(\{p\})=d(p)+1$.  Along a shortest path
$p=p_0\to p_1\to\dots\to p_d\in F$ with $p_{i+1}\cdot x_i=p_i$ the states
$p_0,\dots,p_{d-1}$ are thin, so $\{p_i\}x_i^{-1}=\{p_{i+1}\}$ exactly,
and a letter $y$ with two preimages at $p_d$ gives
$|\{p\}(yx_{d-1}\cdots x_0)^{-1}|\ge2$.  Conversely, if
$|\{p\}u^{-1}|\ge2$, the preimages of $\{p\}$ under the successive
suffixes of $u$ are singletons joined by edges of the digraph up to the
first one of size at least $2$, whose predecessor is fat, so
$d(p)\le|u|-1$.

Hence $\{q\}$ stuck means $d(q)\ge n-1$.  A shortest path has distinct
vertices, so $d(q)=n-1$, the path $q=p_{n-1}\to\dots\to p_0$ exhausts
$Q$, $d(p_i)=i$, $F=\{p_0\}$, and every other state is thin.  Since
$p'\cdot x=p$ forces $d(p)\le d(p')+1$, every letter maps $p_j$ to some
$p_i$ with $i\le j+1$.  For $1\le i\le n-1$ put
$A_i=\{p_i,\dots,p_{n-1}\}$; the only transitions into $A_i$ from outside
come from $p_{i-1}$, with total weight $m_i>0$, and for $p\in A_i$ let
$u_p$ be the total weight of the letters taking $p$ out of $A_i$,
$U_i=\sum_{p\in A_i}u_p$.  A thin state has weighted in-degree
$\sum_x\pi_x|\{p\}x^{-1}|\le1$, so the total weight of the transitions
into $A_i$, which is $(|A_i|-U_i)+m_i$, is at most $|A_i|$: $m_i\le U_i$.
Summing the stationarity equations $e_p=\sum_x\pi_x\,e(\{p\}x^{-1})$ over
$p\in A_i$ gives $e(A_i)=\sum_{p\in A_i}(1-u_p)e_p+m_ie_{p_{i-1}}$, that
is $m_ie_{p_{i-1}}=\sum_{p\in A_i}u_pe_p\ge U_i\min_{A_i}e\ge
m_i\min_{A_i}e$.  So $e_{p_{i-1}}\ge\min_{A_i}e$, and downward induction
from $i=n-1$ gives $e_q=\min_Qe$.

Then $e_q\le e(Q)/n$, with equality only if $e$ is constant, that is, only
if every weighted in-degree equals $1$.  In that case
$\sum_x\pi_x|Tx^{-1}|=|T|$ for every $T\subseteq Q$, so for a stuck $S$
induction on the length shows $|Su^{-1}|=|S|$ for every $|u|\le n-1$,
which is~\eqref{eq:flat}, and Lemma~\ref{lem:stabilise} with a reset
word contradicts it as in the proof of Theorem~\ref{thm:main}.  Hence
$e_q<e(Q)/n$, and in
the integer normalisation $B_e(\{q\})=ne_q-e(Q)\le-1$.
\end{proof}

The least letter weight at which a stuck subset has $B_e(S)\ge0$ is
measured with the letter weights $(t,1-t)$ of
Section~\ref{sec:stationary}: for each stuck subset $S$, $t^{\ast}(S)$ is
the least $t\ge1/2$ at which $n\,e_t(S)\ge|S|\,e_t(Q)$, either letter
being allowed to carry the weight $t$.  Over every stuck subset of every synchronizing strongly
connected binary automaton in the quotient convention, exhaustively at
$n=5$, $6$ and $7$ (32,588, 1,122,529 and 42,605,958 automata with
11,042, 186,497 and 2,827,614 stuck subsets), the minimum of $t^{\ast}$
is exactly $2/3$, and no pair has $t^{\ast}$ in $(1/2,2/3)$.  The minimum
is attained by 9, 8 and 116 pairs (automaton, subset, letter of weight
$t$), which fall into one, one and four isomorphism classes.  The two-letter automata
underlying the first two witnesses above are two of these six classes.
In every attaining pair the letter of weight $1-t^{\ast}$ is a permutation
and the letter of weight $t^{\ast}$ has deficiency $1$ or $2$, and for each class
$n\,\hat e_t(S)-|S|\,\hat e_t(Q)$, computed with the matrix-tree
normalisation of $e_t$, is $(3t-2)$ times a monomial $c\,t^{i}(1-t)^{j}$.
In particular, with pairwise distinct letters and uniform weights, no
stuck subset has $B_e(S)\ge0$ on the exhaustive binary populations at
$n\le7$, on the strata of the binary $n=8$ and $n=9$ populations whose
first letter is a permutation (18,609,570 and 278,959,029 stuck
subsets), or on the exhaustive ternary population at $n=5$ (Table~A.2).
The largest value of $B_e$ on a stuck subset is $-1$ on each of them.  The threshold
is not $2/3$ at every size.  On the stratum of the binary $n=8$
population whose first letter is a permutation, the automaton with
$a=(0\,1)(2\,3)(4\,5)(6\,7)$ and $b=[0,2,3,4,2,6,1,0]$ has the stuck subset
$S=\{4,5\}$, with $\mathrm{minext}(S)=8$, and with weight $t$ on $b$ its
excess $8\,e_t(S)-2\,e_t(Q)$ is a positive multiple of $-(t^{2}-4t+2)$, so
$8\,e_t(S)\ge2\,e_t(Q)$ from $t=2-\sqrt2\approx0.586$ on.  Over the whole
stratum (123,014,054 automata, 18,609,570 stuck subsets) this is the
least threshold: 416 pairs lie below $2/3$, 384 of them at $2-\sqrt2$ and
32 at the root $0.603$ of $2t^{3}-4t^{2}+5t-2$, and none has
$t^{\ast}=1/2$.  On the same stratum at $n=9$, on the grid of step $1/400$
with the pairs at or below the $2/3$ bin refined exactly, the least
threshold is the root $0.594$ of $11t^{2}-20t+8$, attained by 384 pairs,
with 32 pairs at the root $0.643$ of $t^{3}+8t^{2}-18t+8$ and 1,568 at
$2/3$, and again none has $t^{\ast}=1/2$.

\begin{proof}[Proof of Lemma~\ref{lem:twovalued}]
Write $W=\pi_a+\pi_b$ and $v$ for the vector.  Since $b$ is a permutation,
$b(Y)\sqcup b(Z)=Q$.  A state $r$ receives, under $a$, one preimage in $Y$
if $r\in b(Y)$ and none otherwise, and one preimage in $Z$ if $r\in Z$ and
none otherwise; under $b$ it receives its unique preimage $b^{-1}(r)$,
which lies in $Y$ if $r\in b(Y)$ and in $Z$ if $r\in b(Z)$.  Hence
\begin{align*}
W\,(vP)(r)&=\pi_a\bigl(\pi_b[r\in b(Y)]+W[r\in Z]\bigr)
+\pi_b\bigl(\pi_b[r\in b(Y)]+W[r\in b(Z)]\bigr)\\
&=W\bigl(\pi_b+\pi_a[r\in Z]\bigr)=W\,v(r).
\end{align*}
\end{proof}

\begin{proof}[Proof of Proposition~\ref{prop:half}]
In both cases strong connectivity, synchronization (the words
$bbabaaababaaabababb$ and $abbbababbbabababbbababa$ send every state to
$5$, respectively $7$), $\mathrm{minext}(S)$ (every word of length at most
$n-1$ has $|Su^{-1}|\le4$; $bbababbab$, respectively $ababababab$, gives
$5$) and the integer eigenvector of (i) were checked by exhaustive
enumeration and exact rational arithmetic in four implementations written
independently.  For (ii) apply Lemma~\ref{lem:twovalued} with
$Y=\{0,3,6,8,9\}$ and $Z=\{1,2,4,5,7\}$: $b$ is a permutation, $a$ fixes
$2$, $4$ and $7$ and swaps $1$ and $5$, and $a(Y)=\{0,3,4,7,8\}=b(Y)$; the
subset $S$ has two states in each block.
\end{proof}

In (i) neither letter is a permutation and at every length from $1$ to $8$
some word shrinks the subset; in (ii) the subset is flat to depth $7=n-3$,
then loses a state on exactly one word of length $8$ and on four words of
length $9$, never more than one.  

The certified regions of $B$ and $B_e$ are incomparable: on the
exhaustive binary $n=5$ population, $\{B_e\ge0\}$ holds 17,037
proper-subset instances that $\{B\ge0\}$ misses and misses 9,793 that it
holds, and the same happens at $n=4$.  The sweep behind this and the next
statements ran over the exhaustive synchronizing strongly connected
populations at $n=4$, $5$ and $6$ binary and $n=4$ ternary, 76,125,984
proper-subset instances in all.  Moreover $\{B_e\ge0\}$ is not contained in
$C^{\sharp}$, so no positive reweighting of $\sigma_1,\dots,\sigma_{n-1}$
contains it.  A four-state witness settles the last point.  On
$Q=\{0,1,2,3\}$ let $a$ send every state to $0$ and $b$ send $0,1,2,3$ to
$1,2,3,2$; the automaton is synchronizing and strongly connected, with
$e=(6,3,2,1)$.  The subset $S=\{1\}$ has $\sigma_1(S)=-1$ and
$\sigma_2(S)=\sigma_3(S)=0$, so $S\in H^{\sharp}$ and no functional of
Proposition~\ref{prop:region} certifies it, while
$\varepsilon_1=4\cdot3-12=0$, so $B_e(S)=0\ge0$; and $\mathrm{minext}(S)=2$,
by the word $ab$, which sends every state to $1$.  At $n=4$, 115 of the
7,460 proper-subset instances in $H^{\sharp}$ have $B_e\ge0$, none of them
stuck.  On the family of \ref{sec:bounded-deviation} the maximising
subsets have $B_e=-13$ at $n=5$, falling with $n$, so the family
on which no constant per-step bound holds contains no witness either.

\section{Measurements for the second-moment test}
\label{app:qcert}

This appendix collects the membership test, the cost and the measured
shares behind Section~\ref{sec:qcert}.

The attainable second moments are sparse, which is why the threshold of
(Q-CERT+) is not the sharpest test on the same data.  Since
$x^{2}\equiv x$ modulo $2$, the second moment
has the parity of $\sigma_t(S)$; at $k^{t}=4$, $|S|=3$ and $-\sigma_t(S)=3$
the values attainable with every $x_u\le0$ are $3$, $5$ and $9$, so the
observation $7$, below the threshold $9$, refutes the possibility
that every $x_u\le0$.  Moving one unit from a part of size $a$ to a part
of size $b<a$ lowers the second moment by $2(a-b-1)$, so the values
attainable below the threshold $q\,|S|^{2}+r^{2}$ are spaced by even gaps
that depend on which moves the bounds $0\le|Su^{-1}|\le|S|$ allow; in the
cell above the largest attainable value below the threshold falls short of
it by $9-5=4=2(|S|-1)$.  The membership
test on $\bigl(|S|,\sigma_t(S),\sum_{|u|=t}x_u^{2}\bigr)$ is
pseudo-polynomial in $k^{t}$ and is not adopted here.  On the exhaustive
synchronizing strongly connected binary populations, one row per proper
nonempty subset and length $t\le n-1$, it certifies 0 rows of 708 at $n=3$
that the threshold misses, 80 of 52,080 at $n=4$, 15,053 of 3,910,560 at
$n=5$ and 2,083,316 of 347,983,990 at $n=6$, that is 0, 0.15, 0.38 and
0.60 percent, rising at each size; $n\ge7$ is untested.

The computation is a pair count.  From the automaton on $Q\times Q$
carrying the diagonal action of $\Sigma$,
\[\begin{gathered}
\sum_{|u|=t}x_u^{2}=\mathrm{pairdeg}_t(S)-2|S|\,\mathrm{indeg}_t(S)+k^{t}|S|^{2},\\
\mathrm{pairdeg}_t(S)=\sum_{p,q\in S}\mathrm{pairdeg}_t(p,q),
\end{gathered}\]
where $\mathrm{pairdeg}_t(p,q)$ is the $t$-step in-degree of the state
$(p,q)$ in that automaton.  This costs a factor $n$ and the linearity:
$O(kn^{3})$ for all lengths $t\le n-1$ against $O(kn^{2})$ for
$\beta^{\ast}$, and $O(2^{n}n)$ over all subsets against $O(2^{n})$ for $B$.
At that price, in proper-subset instances, the fraction of $\{B<0\}$
removed on each of the eleven populations is given in Table~D.1.

\begin{center}
\renewcommand{\baselinestretch}{1}\footnotesize
\begin{tabular}{@{}llrrr@{}}
\toprule
population & convention & $C^{\sharp}$ (Section~\ref{sec:csharp}) & (Q-CERT) & (Q-CERT+) \\
\midrule
binary $n=3$ & quotient & 0.0000\% & 60.14\% & 60.14\% \\
binary $n=4$ & quotient & 3.4179\% & 79.29\% & 82.07\% \\
binary $n=4$, $d\le2$ & quotient & 0.6422\% & 77.09\% & 80.65\% \\
binary $n=5$ & quotient & 5.7934\% & 86.83\% & 88.98\% \\
binary $n=5$, $d\le2$ & quotient & 1.8838\% & 83.44\% & 86.33\% \\
binary $n=6$ & quotient & 9.0504\% & 90.03\% & 91.44\% \\
binary $n=6$, $d\le2$ & quotient & 3.6524\% & 86.74\% & 88.80\% \\
binary $n=7$, $d\le2$ & quotient & 4.8300\% & 89.79\% & 91.23\% \\
ternary $n=4$ & multiset & 3.0631\% & 90.04\% & 90.64\% \\
ternary $n=4$, $d\le2$ & ordered & 0.6118\% & 90.03\% & 90.81\% \\
ternary $n=5$, $d\le2$ & ordered & 1.0314\% & 95.24\% & 95.51\% \\
\bottomrule
\end{tabular}
\end{center}

\noindent{\small\textbf{Table D.1.}\ What the second moment certifies of
$\{B<0\}$.  For each of the eleven populations of Section~\ref{sec:qcert}
(Table~A.2), the fraction of the proper-subset instances with $B(S)<0$ on
which each test establishes $\mathrm{minext}(S)\le n-1$.  Binary
populations are in the quotient convention of \ref{sec:conventions}; the
ternary $n=4$ population counts the two letters after the first as a
multiset, as \ref{sec:conventions} prescribes, while the two ternary
strata were enumerated with those letters ordered, the only convention in
which they were run.  The $C^{\sharp}$ column is the series of
Section~\ref{sec:csharp} at the precision measured.  Populations at
different sizes and deviations are separate and are not to be merged.\par}
The three columns differ by an order of magnitude.  (Q-CERT) removes
between 9.95 and 147.16 times as much of $\{B<0\}$ as $C^{\sharp}$ on
every population where the $C^{\sharp}$ share is nonzero, the least ratio at
$n=6$ over all deviations and the greatest at ternary $n=4$ on $\{d\le2\}$,
while at $n=3$, where $C^{\sharp}=\{B\ge0\}$ and the enlargement does
nothing, it still certifies 60.14 percent.  So the second moment is not a
refinement of $C^{\sharp}$.  Nor does the gap close at $k=3$: on the exhaustive ternary $n=4$
population (Q-CERT)'s share moves from 79.29 percent at $k=2$ to 90.04 at
$k=3$ and (Q-CERT+)'s from 82.07 to 90.64, while the $C^{\sharp}$ gain moves
from 3.42 to 3.06, so the gap is wider at $k=3$.

The second inclusion of the display of Section~\ref{sec:qcert} is strict, because extending in exactly
$t$ steps and extending within $t$ steps are different conditions.  On the
automaton at $n=4$ in which $a$ sends $0$, $1$ and $2$ to $0$ and $3$ to
$1$, and $b$ sends $0$ to $2$, $1$ to $3$, $2$ to $3$ and $3$ to $0$, which
is synchronizing and strongly connected, the subset $S=\{3\}$ at $t=2$ has
$x_u$ equal to $-1$, $-1$, $0$, $0$ and $\mathrm{minext}(S)=1$.  It is the
union over $t\le n-1$ that recovers the subsets extending within $n-1$.
Lemma~\ref{lem:monotone} settles where the gap can occur.

\begin{lemma}
\label{lem:monotone}
If $d\le2(k-1)$ then for every $S\subseteq Q$ the maximum of $|Su^{-1}|$
over the words $u$ of length $\ell$ is non-decreasing in $\ell$; in
particular extending within $t$ implies extending at length exactly $t$,
and the gap is empty on $\{d\le2(k-1)\}$.
\end{lemma}

\begin{proof}
For any $A\subseteq Q$, grouping the pairs $(q,x)$ with $q\cdot x\in A$ by
letter and by state gives
$\sum_{x\in\Sigma}|Ax^{-1}|=k|A|+\sum_{q\in A}w_q\ge k|A|-d/2$, and if every
letter had $|Ax^{-1}|\le|A|-1$ the sum would be at most $k|A|-k$.  So
$d\le2(k-1)$ leaves a letter with $|Ax^{-1}|\ge|A|$; applied to a maximising
$A=Su^{-1}$ at length $\ell$, this is the statement at length $\ell+1$.
\end{proof}

Since $d$ is even, a gap needs $d\ge2k$, a value attained at both
alphabet sizes measured.  The witness above has $d=4$ with $k=2$.  Over
the exhaustive ternary population at $n=4$ the gap instances number 1,952,
1,157 and 125 at $d=6$, $8$ and $10$, and none lie below.  On the binary
sweeps the gap is empty on every $\{d\le2\}$ population and every instance
found lies at $d\ge4$.  

\section{Bounded deviation and the per-step constant}
\label{sec:bounded-deviation}

No result of Sections~\ref{sec:certificate} to~\ref{sec:sharpness} uses a
degree hypothesis.  A degree hypothesis bounds instead the per-step
constant of the extension method; the family below shows that no
constant suffices.
The Kisielewicz--Szyku\l{}a family \cite{ks} (a
permutation letter with two cycles of coprime lengths $m-1$ and $m$, and a
second letter whose image has $m+1$ states, on $n=2m-1$ states) has
deviation $d=n-3$, and the maximum of $\mathrm{minext}$ over its subsets of
size at least $2$ and at most $n-1$, attained at a single subset of size
$m-1$, is given in Table~E.1:

\begin{center}
\renewcommand{\baselinestretch}{1}\normalsize
\begin{tabular}{@{}lcccccc@{}}
\toprule
$n$ & 5 & 7 & 9 & 11 & 13 & 15 \\
\midrule
$d=n-3$ & 2 & 4 & 6 & 8 & 10 & 12 \\
$\mathrm{minext}$ & 8 & 14 & 22 & 31 & 44 & 56 \\
$2n$ & 10 & 14 & 18 & 22 & 26 & 30 \\
\bottomrule
\end{tabular}
\end{center}

\textbf{Table E.1.} The Kisielewicz--Szyku\l{}a family at $n=2m-1$. Deviation
$d$, the maximum of $\mathrm{minext}$ over subsets of size at least $2$ and
at most $n-1$, and $2n$ for comparison. The maximum is below $2n$ at
$d=2$, equal to it at $d=4$, and above it from $d=6$ on.

The description names the family but does not determine the second
letter; the entries are $\mathrm{minext}$ of the maximising subset on the
family as \cite{ks} defines it, agreeing with the length of the
extending word produced by the strategy of \cite{ks}, $m^{2}-m+2$ for odd
$m$ and $m^{2}-3m/2+4$ for even $m$, that is $(n^{2}+7)/4$ for
$n\equiv1\pmod 4$ and $(n^{2}-n+14)/4$ for $n\equiv3\pmod 4$.

The family refutes the statement that every subset of size at least $2$
and at most $n-1$ extends within $2n$: it can hold at best on $\{d\le4\}$,
since it fails on $\{d\le6\}$.  On $\{d\le4\}$ an exhaustive test attains
equality:
over the full $\{d\le4\}$ population at $n=7$, 22,079,307 synchronizing
strongly connected binary automata with every proper subset measured, the
maximum of $\mathrm{minext}$ over subsets of size at least 2 is $14=2n$;
the same sweep at $n=5$ and $n=6$ (28,968 and 808,403 automata) gives a
maximum of $2n-1$ at $n=5$ and $2n$ at $n=6$.  This is a measurement at three
sizes; $n\ge8$ is untested.

The maxima of the reset threshold are attained in the class $\{d\le2\}$:
the maxima $9,16,25$ over the ambient populations at $n=4,5,6$, and $36$
over the $\{d\le4\}$ population at $n=7$, are all attained at $d=2$, by the
\v{C}ern\'y automata $C_n$ \cite{ce} ($a$ the cycle
$0\to1\to\dots\to n-1\to0$ on $Q=\{0,\dots,n-1\}$, $b$ sending $0$ to
$1$ and fixing every other state), of reset threshold $(n-1)^{2}$.  The
maxima of $\mathrm{minext}$ are not: at $n=5$ the maximum of
$\mathrm{minext}$ over subsets of size 3 is larger on the stratum $d=4$
than on $d=2$, and at $n=7$ the maximum over subsets of size $n-1$ is larger
on $d=4$ than on the whole class $\{d\le2\}$.


\end{document}